\documentclass[12pt]{article}
\usepackage{amsfonts}
\usepackage{amssymb, amsthm, amsbsy, amscd, amsfonts, latexsym, euscript,epsf,stackrel}
\usepackage{amscd}
\usepackage{amsmath}
\usepackage{authblk}

\usepackage{tikz-cd} 

\usepackage[utf8]{inputenc}
\usepackage[T2A]{fontenc}
\usepackage{xcolor}

\usepackage{graphicx}
\graphicspath{{pics/}}

\def\bea{\begin{eqnarray}}
\def\eea{\end{eqnarray}}
\def\nn{\nonumber}

\def\R{\mathbb{R}}
\def\C{\mathbb{C}}
\def\Z{\mathbb{Z}}

\newcommand{\cA}{\mathcal A}

\newcommand{\cC}{\mathcal C}
\newcommand{\cAA}{  A}

\newcommand{\cB}{\mathcal B}

\newcommand{\cI}{\mathcal I}

\newcommand{\cM}{\mathcal M}

\newtheorem{thm}{Theorem}[section]

\newtheorem{lemma}[thm]{Lemma}

\newtheorem{prop}[thm]{Proposition}
\theoremstyle{definition}
\newtheorem{defn}[thm]{Definition}
\theoremstyle{rem}
\newtheorem{rem}[thm]{Remark}

\DeclareMathOperator{\LHS}{LHS}
\DeclareMathOperator{\RHS}{RHS}

\DeclareMathOperator{\Span}{span}

\DeclareMathOperator{\lm}{lm}

\DeclareMathOperator{\T}{T}

\title{Quantisation ideals, canonical parametrisations of the unipotent group and consistent integrable systems}
\author[1,2]{M.~A. Chirkov\footnote{mikhlchirkov@gmail.com}}
\author[3]{A.~V. Mikhailov\footnote{A.V.Mikhailov@leeds.ac.uk}}
\author[1,4,5]{Dmitry V. Talalaev\footnote{dtalalaev@yandex.ru}}
\affil[1]{\small Demidov Yaroslavl State University, Yaroslavl, Russia, 150003, Sovetskaya Str. 14
}

\affil[2]{\small HSE University, Moscow, Russia
}

\affil[3]{\small School of Mathematics, University of Leeds, Leeds, UK
}

\affil[4]{\small Lomonosov Moscow State University, 
119991, Moscow, Russia.
}

\affil[5]{\small Kurchatov Institute -- ITEF, Bolshaya Cheremushkinskaya str. 25, 117218, Moscow, Russia.
}

\begin{document}
\maketitle

\abstract{Using the method of quantisation ideals, we construct a family of quantisations corresponding to Case $\alpha$ in Sergeev’s classification of solutions to the tetrahedron equation. This solution describes transformations between special parametrisations of the space of unipotent matrices with noncommutative coefficients. We analyse the classical limit of this family and construct a pencil of compatible Poisson brackets that remain invariant under the re-parametrisation maps (mutations).
This decomposition problem is closely related to Lusztig’s framework, which makes links with the theory of cluster algebras.
Our construction differs from the standard family of Poisson structures in cluster theory; it provides  deformations of  log-canonical brackets. Additionally, we identify a family of integrable systems defined on the parametrisation charts, compatible with mutations.
}
\tableofcontents

%\section{Введение}
%Метод идеалов квантования, предложенный в работах \cite{M}, \cite{BM}

\section{Introduction}
\label{Intro}
%Решение Люстига, источник, задача.
\paragraph{Quantisation ideals.}
The quantisation ideals approach was originally developed for differential and differential-difference dynamical systems \cite{M}. It differs significantly from the conventional quantisation framework. Traditionally, one starts with a classical dynamical system defined on a commutative algebra of functions, and quantisation is viewed as a deformation of the commutative multiplication into a non-commutative, associative product that in the classical limit reproduces the Hamiltonian structure of the system.

In contrast, the quantisation ideals method is applied to systems defined on a free associative algebra. The central idea is to identify two-sided ideals that are invariant under the dynamical flow and such that the corresponding factor algebras admit a Poincar\'e--Birkhoff--Witt (PBW) basis. Such an ideal is referred to as a quantisation ideal, and it defines the commutation relations in the resulting quantum algebra.

This method not only recovers quantisations with classical (commutative) limits but also allows for the construction of quantum dynamical systems with no classical counterpart - such as those involving fermionic degrees of freedom. It has been successfully applied to a wide range of integrable systems, including the Volterra chain \cite{CMW22}, stationary KdV flows \cite{BM}, the Euler top in the external field \cite{MS}, as well as the Toda, Kaup and Ablowitz--Ladik systems, among many others. We refer to this framework as the method of quantisation ideals.

In this paper, we extend the method to the setting of discrete dynamics. Our aim is to find possible quantum solutions to the Zamolodchikov tetrahedron equation \cite{Zam}:
\begin{equation}\label{zameqI}
\T_{123} \circ \T_{145} \circ \T_{246} \circ \T_{356} = \T_{356} \circ \T_{246} \circ \T_{145} \circ \T_{123},
\end{equation}
associated with the following invertible polynomial ring homomorphism $\T \,:\, \C[x', y', z']\mapsto\C[x, y, z]$, defined by the map
\bea\label{reparam}
x' &=& x,\nn\\
y' &=& y + xz, \\
z' &=& z \nn,
\eea
which provides a solution to equation (\ref{zameqI}) in the polynomial ring $\C[ x_1,x_2,x_3,x_4,x_5,x_6]$,  where
\begin{equation}\label{TijkI}
\T_{ijk}(x_l) = \begin{cases}
x_j + x_i x_k, l = j, \\
x_l, l \ne j.
\end{cases}
\end{equation}
This solution, along with the corresponding map, appears in Sergeev’s classification \cite{Serg} as Case $\alpha$.

The map $\T$ arises in the decomposition problem for the unipotent group as a product of its one-parameter subgroups. This problem is closely related to Lusztig’s work \cite{Lus} on the positive parametrisation of elements in the unipotent subgroup $N(n,\R)$, where $N(n,\R)$ denotes the group of upper triangular $n \times n$ real matrices with ones on the diagonal.

Let $u_{ij}(t)$ denote the elements of the one-parameter subgroups of $N(n,\R)$:
\begin{equation*}
u_{ij}(t) = I_n + t E_{ij}, \qquad i < j,
\end{equation*}
where $I_n$ is the identity matrix and $E_{ij}$ is the elementary matrix unit with a $1$ in the $(i,j)$-entry and zeros elsewhere.
In the simplest case $n=3$ there are two types of parametrisations of an element $a\in N(3,\R)$:
\bea
\label{deco}
 a=\left(\begin{array}{ccc}
1 & x' & y'\\
0 & 1 & z'  \\
0 & 0 & 1 
\end{array}\right)=u  _{12}(x) u  _{13}(y) u_{23}(z)=u_{23}(z') u  _{13}(y') u  _{12}(x').
\eea
The coordinates of these charts are related by the transformation \eqref{reparam}.

The same formulas remain true if we consider the decomposition problem (\ref{deco}) in the group $N(3,\cAA)$, where  $\cAA=\C\langle x,y,z\rangle$ is the associative unital free algebra generated by noncommutative variables $x,y,z$.  Remarkably, the transformations $\T_{i,j,k}$ (\ref{TijkI}) provide a solution to the Zamolodchikov tetrahedron equation within the free associative algebra $\cA=\C\langle x_1,x_2,x_3,x_4,x_5,x_6\rangle$.

This result follows from the uniqueness of the decomposition of a generic element $a\in N(4,\cA)$
into the product:
\bea
a=u  _{12}(t_1) u  _{13}(t_2) u_{23}(t_3) u  _{14}(t_4) u_{24}(t_5) u_{34}(t_6).\nn
\eea
Details of this construction and its implications are discussed in Section~\ref{sec2}.

Factorisations in the group $N(4, \cA)$ were previously studied in \cite{BR} in the context of noncommutative Bruhat cells. In particular, Section 3.3 of that work presents explicit formulas for recovering the parameters of the factorisation in terms of quasideterminants.

Our first main result is the classification of quadratic quantisation ideals $I \subset\cAA= \C\left<x, y, z\right>$, that is, ideals satisfying the following two conditions:
\begin{itemize}
\item
The ideal $I$ is invariant under the mutation map (\ref{reparam}).
\item
The quotient algebra $\cAA_I=\cAA/I$ admits a PBW basis consisting of normally ordered monomials.
\end{itemize}

These conditions ensure that the resulting quantum algebra $\cAA_I$ has the same polynomial growth as the commutative polynomial ring in three variables, and that the re-parametrisation map \eqref{reparam} is well defined on $\cAA_I$.We show that there are exactly three distinct quantisation ideals satisfying these conditions (Theorem \ref{thm-2-1}). Notably, one of these ideals defines a quantum algebra that remains non-commutative for all choices of quantum parameters, meaning that the corresponding quantum system does not admit any classical (commutative) limit.

Next, we study quantisation ideals associated with the unipotent group $N(4, \cA)$. Specifically, we identify two-sided ideals $I \subset \cA = \C\left<x_1, \ldots, x_6\right>$ that satisfy the following conditions:
\begin{itemize}
\item
The ideal $I$ is stable under all maps $\T_{i,j,k}$ appearing in the Zamolodchikov equation (\ref{zameqI}).
\item
The quantum algebra $\cA_I = \cA / I$ admits a PBW basis.
\end{itemize}
Solving the classification problem for triangular quantisation ideals of generic type, we find four essentially distinct ideals. In addition, we construct an explicit example of a non-generic quantisation ideal.
All of these quantisation ideals are homogeneous deformations of the toric ideal, and the corresponding quantum algebras admit a classical (commutative) limit.

We conclude the paper by studying the classical limit of this family of quantisations. In this limit, we construct a pencil of compatible Poisson brackets on the polynomial ring $\C\left[x_1, \ldots, x_6\right]$ that are preserved by the maps $\T_{ijk}$. We also identify a corresponding family of integrals of motion for the induced dynamics. This leads to a class of integrable systems on the unipotent group $N(4, \C)$, which are consistent with the mutation maps defined by the Zamolodchikov tetrahedron equation.
%The quantisation task in the case of the Hirota map was solved in \cite{DK}.

%\paragraph{Noncomutative Lusztig's problem}
%\label{NCLuz}

\section{Charts of the Unipotent group and symmetries}\label{sec2}
\subsection{Parameterisations of the unipotent group}
\label{sec_param}
We consider the problem of parametrisation of  $M\in Mat_4(\cA )$ of the form
\bea\label{M4}
M=\left(\begin{array}{cccc}
1 & x_1 & x_2 & x_4\\
0 & 1 & x_3 & x_5\\
0 & 0 & 1 & x_6\\
0 & 0 & 0 & 1
\end{array}\right),
\eea
by the one-parameter subgroups, generated by $u_{ij}(t)$ from the Introduction.
Using the re-parametrisation $T_{ijk}$ we get the following sequence of decompositions
\bea
M&=&u_{34}(x_6) u_{24}(x_5) u  _{14}(x_4){\color{blue} u_{23}(x_3) u  _{13}(x_2) u  _{12}(x_1)}\nn\\
&=&u_{34}(x_6) {\color{magenta}u_{24}(x_5) u  _{14}(x_4) u  _{12}(x'_1)} u  _{13}(x'_2) u_{23}(x'_3)\nn\\
&=&{\color{green}u_{34}(x_6) u  _{12}(x''_1)} u  _{14}(x'_4) {\color{red}u_{24}(x'_5) u  _{13}(x'_2)} u_{23}(x'_3)\nn\\
&=&u  _{12}(x''_1){\color{blue}u_{34}(x_6)  u  _{14}(x'_4) u  _{13}(x'_2)} u_{24}(x'_5) u_{23}(x'_3)\nn\\
&=&u  _{12}(x''_1)u  _{13}(x''_2)  u  _{14}(x''_4) {\color{magenta}u_{34}(x'_6) u_{24}(x'_5) u_{23}(x'_3)}\nn\\
&=&u  _{12}(x''_1)u  _{13}(x''_2)  u  _{14}(x''_4) u_{23}(x''_3)  u_{24}(x''_5) u_{34}(x''_6).
\label{LHS}
\eea
The transitions from the first line to the second and from the second to the third correspond to the application of the inverse transformation $T_{123}^{-1}$ and $T_{145}^{-1}$, respectively. The transition from the third line to the fourth corresponds to the permutation of the commuting generators of the one-parameter subgroups. The transitions from the 4th line to the fifth and from the fifth to the sixth are $T_{246}^{-1}$ and $T_{356}^{-1}$, respectively.
Similarly, applying the corresponding transformations in a different order we obtain
\bea
M&=&{\color{blue} u_{34}(x_6) u_{24}(x_5) u_{23}(x_3)}u  _{14}(x_4)  u  _{13}(x_2) u  _{12}(x_1)\nn\\
&=& u_{23}(x^*_3) u_{24}(x^*_5) {\color{magenta}u_{34}(x^*_6)u  _{14}(x_4)  u  _{13}(x_2)} u  _{12}(x_1)\nn\\
&=&u_{23}(x^*_3) {\color{green} u_{24}(x^*_5) u  _{13}(x^*_2)} u  _{14}(x^*_4) {\color{red} u_{34}(x^{**}_6)  u  _{12}(x_1)}\nn\\
&=&u_{23}(x^*_3) u  _{13}(x^*_2) {\color{blue} u_{24}(x^*_5) u  _{14}(x^*_4) u  _{12}(x_1)} u_{34}(x^{**}_6) \nn\\
&=&{\color{magenta}u_{23}(x^*_3) u  _{13}(x^*_2) u  _{12}(x^*_1)}  u  _{14}(x^{**}_4) u_{24}(x^{**}_5) u_{34}(x^{**}_6) \nn\\
&=&u  _{12}(x''_1)u  _{13}(x''_2)  u  _{14}(x''_4) u_{23}(x''_3)  u_{24}(x''_5) u_{34}(x''_6).
\label{RHS}
\eea
As a result, we get 8 different parameterisations of unipotenr matrices $M$: $\{x_1,x_2,\ldots,x_6\},$ $\{x'_1,x'_2,x'_3,x_4,x_5,x_6\},$ etc., as well as parameterisations in the second series of transformations: $\{x_1,x_2,x^*_3,x_4,x^*_5,x^*_6\},$ etc. They are presented on Fig. \ref{fig_graph}.
\begin{figure}[h!]
\label{fig_graph}
\center
\includegraphics[]{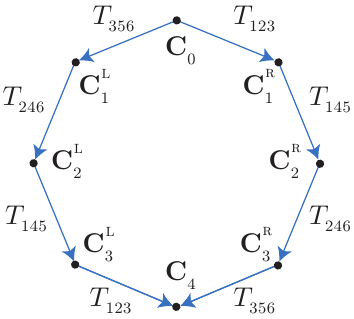}
\caption{The graph of parametrisation charts}
\end{figure}
The parametrisation $C_1^R$, for example, corresponds to $\{x'_1,x'_2,x'_3,x_4,x_5,x_6\},$ and $C_3^L$ to $\{x_1^*,x_2^*,x_3^*,x_4^{**},x_5^{**},x_6^{**}\}.$
The last parametrisations in both chains coincide due to the tetrahedron equation. 
%\begin{remark}
%We interpret these parametrisations as a non-commutative analogue of positive maps on the space of upper triangular matrices. Note that all parametrisations are isomorphic as algebras. This consideration distinguishes our parametrisations from quantum cluster algebras.
%\end{remark}

\subsection{Involution and other orders}
\label{inv}
Recall that for the matrix ring with coefficients in the associative algebra there is the following homomorphism
\bea
\theta:Mat_n(\mathcal{A})\rightarrow Mat_n(\mathcal{A}^{op});\qquad \theta(M)^i_j=M^{n-j+1}_{n-i+1};\nn
\eea
this is the reflection of the matrix with respect to the antidiagonal. It should also be noted that $\mathcal{A}^{op}$ is the same vector space as the algebra $\mathcal{A}$, but with inverse multiplication:
\bea
a \circ_{op} b = b \circ a.\nn
\eea
We apply the homomorphism $\theta$ to the left-hand side of the tetrahedron equation \eqref{LHS} 
\bea
\theta(M)&=&u  _{12}(x_6) u  _{13}(x_5) u  _{14}(x_4) u_{23}(x_3) u_{24}(x_2) u_{34}(x_1)\nn\\
&=&u  _{12}(x_6) u  _{13}(x_5) u  _{14}(x_4) u_{34}(x'_1) u_{24}(x'_2) u_{23}(x'_3)\nn\\
&=&u  _{12}(x_6) u_{34}(x''_1) u  _{14}(x'_4) u  _{13}(x'_5) u_{24}(x'_2) u_{23}(x'_3)\nn\\
&=&u_{34}(x''_1) u  _{12}(x_6)  u  _{14}(x'_4) u_{24}(x'_2) u  _{13}(x'_5) u_{23}(x'_3)\nn\\
&=&u_{34}(x''_1)u_{24}(x''_2)  u  _{14}(x''_4) u  _{12}(x'_6) u  _{13}(x'_5) u_{23}(x'_3)\nn\\
&=&u_{34}(x''_1)u_{24}(x''_2)  u  _{14}(x''_4) u_{23}(x''_3)  u  _{13}(x''_5) u  _{12}(x''_6).
\label{LHST}
\eea
Now let's make a substitution:
\bea
\tau: x_1\leftrightarrow y_6;~x_2\leftrightarrow y_5;~x_3\leftrightarrow y_3;~x_4\leftrightarrow y_4;~x_5\leftrightarrow y_2;~x_6\leftrightarrow y_1.
\label{tau}
\eea
\bea
\tau\circ\theta(M)&=&u  _{12}(y_1) u  _{13}(y_2) u  _{14}(y_4) u_{23}(y_3) u_{24}(y_5) u_{34}(y_6)\nn\\
&=&u  _{12}(y_1) u  _{13}(y_2) u  _{14}(y_4) u_{34}(y'_6) u_{24}(y'_5) u_{23}(y'_3)\nn\\
&=&u  _{12}(y_1) u_{34}(y''_6) u  _{14}(y'_4) u  _{13}(y'_5) u_{24}(y'_5) u_{23}(y'_3)\nn\\
&=&u_{34}(y''_6) u  _{12}(y_1)  u  _{14}(y'_4) u_{24}(y'_5) u  _{13}(y'_2) u_{23}(y'_3)\nn\\
&=&u_{34}(y''_6)u_{24}(y''_5)  u  _{14}(x''_4) u  _{12}(y'_1) u  _{13}(y'_2) u_{23}(x'_3)\nn\\
&=&u_{34}(y''_6)u_{24}(y''_5)  u  _{14}(y''_4) u_{23}(y''_3)  u  _{13}(y''_2) u  _{12}(y''_1).
\label{LHSTr}
\eea
A similar expression for the right-hand side of the series of decompositions \eqref{RHS} takes the form
\bea
\tau\circ\theta(M)&=&u  _{12}(y_1) u  _{13}(y_2) u_{23}(y_3) u  _{14}(y_4)  u_{24}(y_5) u_{34}(x_6)\nn\\
&=& u_{23}(y^*_3) u_{24}(y^*_2) u  _{12}(y^*_1)u  _{14}(y_4)  u_{24}(y_5) u_{34}(y_6)\nn\\
&=&u_{23}(y^*_3)  u  _{13}(y^*_2) u_{24}(y^*_5) u  _{14}(y^*_4) u  _{12}(y^{**}_1)  u_{34}(y_6)\nn\\
&=&u_{23}(y^*_3) u_{24}(y^*_5) u  _{13}(y^*_2) u  _{14}(y^*_4) u_{34}(y_6) u  _{12}(y^{**}_1) \nn\\
&=&u_{23}(y^*_3) u_{24}(y^*_5) u_{34}(y^*_6)  u  _{14}(y^{**}_4) u  _{13}(y^{**}_2) u  _{12}(y^{**}_1) \nn\\
&=&u_{34}(y''_6)u_{24}(y''_5)  u  _{14}(y''_4) u_{23}(y''_3)  u  _{13}(y''_2) u  _{12}(x''_1).
\label{RHSTr}
\eea
These calculations imply that the charts corresponding to the 1st, 2nd, 3rd, 5th and 6th lines of the equations \eqref{LHSTr} and \eqref{RHSTr} are parametrisations in the spaces of unipotent matrices of the form
\bea
M''=\left(\begin{array}{cccc}
1 & y''_1 & y''_2 & y''_4\\
0 & 1     & y''_3 & y''_5\\
0 & 0     & 1     & y''_6\\
0 & 0     & 0     & 1
\end{array}\right).
\eea
In this case, the transformations between the charts from top to bottom are now carried out by the maps $T_{ijk}$ themselves.

\section{ Quantisations of the unipotent groups}\label{section2}

\subsection{The case $N(3,\cAA)$, classification of stable PBW ideals}\label{sec21}
%\subsection{Reparametrisations of $N(3)$, invariant ideals}

Consider a free associative unital algebra $A = \mathbb{C} \left< x_1, x_2, x_3 \right>$ and its automorphism $\T$,  defined on generators by
\begin{equation}\label{Phi}
\T: A \to A,\qquad \T(x_1) = x_1,\quad \T(x_2) = x_2 + x_1 x_3,\quad \T(x_3) = x_3.
\end{equation}
It is known \cite{NCL} that this map $\T$ gives rise to a solution of the tetrahedron equation, analogous to the one in the commutative case.

By {\em quantisation}  we mean the canonical reduction to a quotient algebra $A/I$ over a two sided $\T$--stable ideal $I\subset A$ , i.e.  $\T(I) \subset I$, such that  $A/I$ admits a  basis
\begin{equation}\label{Bx}
 B =\langle  {x}_1^n {x}_2^m {x}_3^k\,|\, n,m,k\in\Z_{\geqslant
0}\rangle
\end{equation}
of normally ordered monomials. 
%This reduction is well-defined only if the ideal $I$ is $\T$--stable, i.e.  $\T(I) \subset I$. 
The ideal  $I$ and quotient algebra $A/I$ are referred to as the {\em quantisation ideal} and the {\em quantum algebra}, respectively.

For brevity, we refer to normally ordered monomials as standard. Non-standard
monomials can be expressed in terms of the standard monomial basis modulo the
ideal. Equality modulo an ideal $I$ will be denoted by
$\stackrel{I}{\equiv}$, or simply $\equiv$ when the ideal under discussion is
clear. A basis of normally ordered monomials (\ref{Bx}) we refer as a Poincar\'e--Birkhoff--Witt or a {\em PBW basis}. An ideal $I\subset A$ we refer as a {\em PBW ideal} if the quotient algebra $ A/I$ admits a PBW basis.

As a candidate for a quantisation ideal, we consider the ideal $I$ generated by the following three polynomials:
\begin{equation}
 \label{candJ3}
 I=\langle F_1:=x_2x_1-f_1,\ F_2:=x_3x_1-f_2,\ F_3:=x_3x_2-f_3\rangle\, ,
\end{equation}
where $f_k$ are general quadratic non-homogeneous polynomials expressed in terms of normally ordered monomials as:
\begin{equation}\label{fk3}
 f_k= x_1^2 \alpha   _{k1}+ x_1  x_2 \alpha   _{k2}+ x_2^2 \alpha   _{k3}+
x_1  x_3 \alpha   _{k4}+ x_2  x_3 \alpha   _{k5}+ x_3^2 \alpha   _{k6}+ x_1
\beta   _{k1}+ x_2 \beta   _{k2}+ x_3 \beta   _{k3}+\gamma _k \, ,
\end{equation}
with $\alpha_{ij},\beta_{ij},\gamma_{i}\in\C$ as arbitrary parameters. We further assume:
\begin{equation}\label{notzero}
\alpha_{12}\ne 0,\quad \alpha_{24}\ne0,\quad \alpha_{35}\ne0.
\end{equation}

The conditions for $\T$--stability of the ideal and the existence of the PBW basis in the quotient algebra $A/I$ impose constraints on these parameters.

\begin{thm}\label{thm-2-1}
An ideal $I$ (as defined in   (\ref{candJ3}), (\ref{fk3}), and (\ref{notzero})) is $\T$--stable and PBW if and only if it is
generated by one of the following sets of polynomials:
\paragraph{Case 1:}
\begin{equation}\label{case1}
\begin{array}{l}
     x_2\,x_1- x_1\,x_2-x_1^2 \alpha _{11}-x_1\,x_3
\alpha _{14}-x_3^2 \alpha _{16}-x_1 \beta _{11}-x_3 \beta
_{13}-\gamma _1 ,\\x_3\,x_1-
x_1\,x_3 ,\\x_3\,x_2- x_2\,x_3-x_1^2 \alpha
_{31}-x_1\,x_3 \alpha _{34}-x_3^2 \alpha _{36}-x_1 \beta
_{31}-x_3 \beta _{33}-\gamma
_3 ;
  \end{array}
\end{equation}

\paragraph{Case 2:}
\begin{equation}\label{case2}
 \begin{array}{l}
x_2\,x_1+x_1\,x_2- x_1^2 \alpha
_{11}-x_3^2 \alpha _{16}-x_3 \beta _{13}-x_1 \beta _{33}-\gamma
_1,\\x_3\,x_1+x_1\,x_3,\\x_3\,x_2+x_2\,x_3- x_1^2 \alpha
_{31}-x_3^2 \alpha _{36}-x_1 \beta _{31}-x_3 \beta _{33}-\gamma
_3  ;
  \end{array}
\end{equation}

\paragraph{Case 3:}
\begin{equation}\label{case3}
 \begin{array}{l}
   x_2\,x_1- \omega\, x_1\,x_2 - x_1^2 \alpha _{11}-x_1 \beta
_{33} ,\\x_3\,x_1- \omega\, x_1\,x_3
,\\x_3\,x_2- \omega\, x_2\,x_3-x_3^2 \alpha _{36}-x_3 \beta
_{33}  ,
  \end{array}
\end{equation}
where $\alpha_{ij}, \ \beta_{ij},\ \gamma_i$ and $\omega\ne0$ are arbitrary parameters.
\end{thm}
\medskip

\begin{rem}\phantom{.}
\begin{itemize}
\item \textbf{Cases 1, 2} can be viewed as a deformations of the commutative polynomial algebra $\C[x_1,x_2,x_3]$.
\item \textbf{Case 2} corresponds to a quantum algebra that has no commutative limit. It can be viewed as a deformation of the noncommutative algebra $ \C\langle x_1,x_2,x_3\rangle / \langle x_1 x_2+x_2 x_1,x_1 x_3+x_3 x_1,x_3 x_2+x_2 x_3\rangle$. The associated Poisson structure can be constructed using the techniques developed in \cite{MV}.
           \end{itemize}
\end{rem}

\medskip
{\bf Proof:} Our proof is based on two lemmas and Levandovskyy's Theorem 2.3 \cite{Levandovskyy}.
Let us first assume that the quotient algebra $A/I$ admits the PBW basis (\ref{Bx}) and find a general form of a $\T$--stable ideal. Then we will find and solve conditions for the existence of the PBW basis.

\begin{lemma}\label{lem1}
 Let the standard monomials be linearly independent in $A/I$.    Then a $\T$--stable ideal $I$ is generated by polynomials of the form
 \begin{equation}
  \label{idJ0}
 \begin{array}{l}
     x_2\,x_1-\omega x_1\,x_2-x_1^2 \alpha _{11}-x_1\,x_3
\alpha _{14}-x_3^2 \alpha _{16}-x_1 \beta _{11}-x_3 \beta
_{13}-\gamma _1 ,\\x_3\,x_1-\omega
x_1\,x_3 ,\\x_3\,x_2- \omega x_2\,x_3-x_1^2 \alpha
_{31}-x_1\,x_3 \alpha _{34}-x_3^2 \alpha _{36}-x_1 \beta
_{31}-x_3 \beta _{33}-\gamma
_3 ,
  \end{array}
 \end{equation}
where the coefficients $\alpha_{ij}, \ \beta_{ij},\ \gamma_i$ and $\omega\ne0$ are arbitrary parameters.
\end{lemma}

\medskip

{\bf Proof of Lemma \ref{lem1}} The stability conditions
$\T(F_k)\stackrel{I}{\equiv} 0,\ k=1,2,3$ yield a system of
polynomial equations for the coefficients of  $F_k$. It follows from the form of polynomials $F_k$ that linear combinations of standard monomials do not belong to the ideal $I$. Consider
\[
 \T(F_2)=F_2-x_1 x_3 x_1 x_3 \alpha _{23}-x_1^2 x_3 \alpha _{22}-x_1 x_3 x_2 \alpha _{23}-x_2 x_1x_3 \alpha _{23}-x_1 x_3^2 \alpha _{25}-x_1 x_3 \beta _{22}.
\]
 Using the relation
\[ x_1 x_3 x_1 x_3\equiv x_1 f_1 x_3=\alpha_{24}x_1^2x_3^2+\cdots \gamma_2x_1x_3\not\in J,\quad \alpha_{24}\ne0,
\]
we conclude
that $\T(F_2)\in I$
implies
\[\alpha _{23}= \alpha _{22}  = \alpha _{25}=\beta_{22}=0,\]
and thus
\begin{equation}\label{f2}
 x_3x_1\equiv f_2=x_1 x_3 \alpha _{24}+x_1^2 \alpha _{21}+x_3^2 \alpha _{26}+x_1 \beta _{21}+x_3 \beta _{23}+\gamma _2.
\end{equation}
Similarly, the conditions $\T(F_1),\T(F_3)\in I$ imply
\[\alpha_{13}=\alpha_{33}=0,\] and therefore
\begin{eqnarray}\label{F1}
 &&\T(F_1)\equiv-x_1^2 x_3 \alpha _{1,2}-x_1 x_3^2 \alpha _{1,5}-x_1 x_3 \beta _{1,2}+x_1 x_3 x_1,\\
 \label{F3}
&& \T(F_3)\equiv-x_1^2 x_3 \alpha _{3,2}-x_1 x_3^2 \alpha _{3,5}-x_1 x_3 \beta _{3,2}+x_3 x_1 x_3.
\end{eqnarray}
We use $x_1 x_3 x_1\equiv x_1 f_2$ and $x_3 x_1 x_3\equiv f_2 x_3$ to present (\ref{F1}), (\ref{F3}) in the basis of standard monomials:
\[\begin{array}{l}
 \T(F_1)\equiv x_1 ^2 x_3 \left(\alpha _{24}-\alpha _{12}\right)+x_1 x_3^2 \left(\alpha _{26}-\alpha _{15}\right)+x_1^3 \alpha _{21}+x_1 x_3 \left(\beta _{23}-\beta _{12}\right)+x_1^2 \beta _{21}+\gamma _2 x_1,\\
 \T(F_3)\equiv x_1 ^2 x_3 \left(\alpha _{21}-\alpha _{32}\right)+x_1 x_3^2 \left(\alpha _{24}-\alpha _{35}\right)+x_3^3 \alpha _{26}+x_1 x_3 \left(\beta _{21}-\beta _{32}\right)+x_3^2 \beta _{23}+\gamma _2 x_3.
  \end{array}
\]
Since we assume that the standard monomials are linearly independent, we get
\[\alpha _{12}=\alpha _{24}=\alpha _{35}:=\omega\ne 0,\quad \alpha _{15}
=\alpha _{21} = \alpha _{26} =\alpha _{32} =\beta _{12} =\beta _{21} =\beta
_{23} =\beta _{32} =\gamma _2= 0.\qquad \Box
\]
In algebra $A/I$ with $I$ generated by polynomials of the form  (\ref{idJ0}) we can choose any ordering of the generators as a normal ordering. Moreover, the conditions for the existence of a PBW basis do not depend on the ordering.
\begin{lemma}\label{lem_ordering}
 Let an ideal $I\subset A$ be generated by polynomials  (\ref{idJ0}) and $\omega\ne 0$. Then the quotient algebra $A/I$ admits a PBW basis
$ B =\langle  {x}_1^n {x}_2^m {x}_3^k\,|\, n,m,k\in\Z_{\geqslant
0}\rangle    $ if and only if it admits a PBW basis $ B_\sigma =\langle  {x}_{\sigma(1)}^n {x}_{\sigma(2)}^m {x}_{\sigma(3)}^k\,|\, n,m,k\in\Z_{\geqslant
0}\rangle    $, where $\sigma$ is any permutation of $\{1,2,3\}$.
\end{lemma}

To determine a criterion for the existence of a PBW basis, we will use Levandovskyy's Theorem \cite{Levandovskyy} (Theorem 2.3). To do so, we introduce the notations and definitions that will be used throughout this and subsequent sections of the paper.

Consider a free associative algebra $\cA=\C\langle s_1,\ldots,s_n\rangle$ with $n$ generators $s_1,\ldots, s_n$. Let $\cM$ denote the set of all monomials in $\cA$.
Define the set of {\em standard monomials}
\begin{equation}\label{setB}
\cB=\{s_1^{\alpha_1}s_2^{\alpha_2}\cdots s_n^{\alpha_n}\,|\, (\alpha_1,\alpha_2,\ldots,\alpha_n)\in\Z_{\geqslant 0}^n\}
\end{equation}
 and let $L_\cB=\Span_\C(\cB)$
 denotes the $\C$--linear space spanned by the standard monomials. Two monomials $a,b\in\cA$ are said to be {\em similar} $a\sim b$ if $b$ can be obtained from $a$ by a permutation of letters (e.g. $s_1^2s_3 \not\sim s_1 s_2s_3\sim s_3s_2s_1\sim s_2s_1s_3$). There is a projection $\pi_0:\, \cM\mapsto\cB$ defined by the condition $\pi_0(a)=\pi_0(b) \, \Leftrightarrow \, a\sim b$.

 \begin{defn}\label{monord}
From this point onward, we fix the following monomial ordering $\prec$ on $\cM$: the generators are ordered as $s_1 \prec \cdots \prec s_n$. For standard monomials, we use the reverse (right) lexicographical order.
 Namely, $s_1^{\alpha_1}\cdots s_n^{\alpha_n}\prec s_1^{\beta_1}\cdots s_n^{\beta_n}$ if in the difference $(\beta_1-\alpha_1,\ldots,\beta_n-\alpha_n)$ the rightmost nonzero entry is positive. We extend it to $\cM$ as following:
\begin{itemize}
 \item   $a\prec b$ if $\pi_0(a)\prec\pi_0(b)$;
 \item if $a\sim b$, we find their largest common left subword $m$ such
that 
$a=ma_1,\, b=mb_1$, or set $m=1$ if no such subword exists. Let $s_i$  and $s_j$ denote the first letters of $a_1$ and $b_1$, respectively (counting from the left). Then, $a\prec b$ if $s_i\prec s_j$.
\end{itemize}
\end{defn}
In this ordering, standard monomials are the smallest within their equivalence classes of monomials.
For example:
\begin{equation}\label{exord}
 1\prec s_1\prec s_1^2\prec s_2\prec s_1s_2\prec s_2s_1\prec s_2^2\prec s_3\prec s_1s_3\prec s_2s_3\prec s_1s_2 s_3\prec s_2s_1s_3\prec s_3s_1s_2 \, .
\end{equation}
Any non-zero $f\in\cA$ can be written uniquely as $f=c\,m+f'$, where $c\in\C^*,\, m\in\cM$ and $m'\prec m$. We define $
\lm(f):=m$ the {\em leading monomial} of  $f$, and refer to the product  $c\,m$ as the {\em leading term} of $f$.  
Suppose we have a set of $n(n-1)/2$ polynomials $f_{ij}$:
\begin{equation}\label{setF}
 \begin{array}{lll}
  F &=&\{f_{ij}\,|\, 1\leqslant i<j\leqslant n\}\subset  \cA=\C\langle s_{1},\ldots , s_{n}\rangle,\ \ \mbox{where}\\
  &&\forall\, i<j\quad f_{ij}=s_j s_i-\omega_{ij}s_i s_j-d_{ij},\ \ \omega_{ij}\in \C^*,\ \ d_{ij}\in L_{\cB},
   \end{array}
\end{equation}
such that
\begin{equation}\label{condF}
  \lm(d_{ij})\prec s_i s_j\prec\lm (f_{ij})=s_j s_i.
\end{equation}
In the polynomials $f_{ij}$, all monomials, except the leading monomial $ s_j s_i$, are normally ordered (standard). Also we consider only homogeneous quadratic polynomials $d_{ij}$. 

Any element  $a\in\cA$ can be brought to normally ordered form modulo the set of polynomials $F$. Indeed, let $m$ be the leading non-standard monomial in $a$. Since $m$ is non-standard, it must contain a pair of generators $s_i s_j$
appearing in the wrong order (i.e., with $i<j$). Thus, we can write
$m=p s_j s_i q$, where $p,q$ are  monomials. We rewrite this as $m=p( s_j s_i-f_{ij})q+pf_{ij}q$. By construction, the leading monomial of $\lm (p( s_j s_i-f_{ij})q)$ is strictly smaller than $m$ with respect to the monomial order $\prec$.
Since the number of monomials of a fixed grading is finite, the process of reordering terminates after finitely many steps.
%Since ``$\prec$'' is a well ordering, after finitely many steps, 
As a result, we can express $a$ as   $a=(:a:)+\tilde{a}$, where   $(:a:)\in L_\cB$  is a normally ordered polynomial (i.e., a linear combination of standard monomials), and the remainder $\tilde{a}$ is a finite sum of the form $\tilde{a}=\sum\limits_{k}\sum\limits_{1\leqslant i<j\leqslant n} p_{kij}f_{ij}q_{kij},\ \ p_{kij},\,q_{kij}\in\cA$. In other words, treating the relations $f_{ij} \stackrel{\left<F\right>}{\equiv} 0$ as rewriting rules, we can transform any element   $a\in\cA$ into a normally ordered form $(:a:)$ in a finite number of steps.

In general, the normally ordered form of an element $a\in\cA$ is not uniquely defined unless additional conditions are imposed on the coefficients of the polynomials $f_{ij}$. These conditions are well known - see for example Theorem 2.1 in \cite{PP}. Here we will use Levandovskyy's Theorem   \cite{Levandovskyy} (Theorem 2.3), which we have adapted to suit our notations.

\begin{thm} \label{LevThm}
Let  $\cI=\langle F \rangle\subset \cA$  be the two--sided ideal    generated by a set of polynomials $F$ (\ref{setF}) satisfying conditions (\ref{condF}). Then, the following conditions are equivalent:
\begin{enumerate}
 \item The quotient $\C$--algebra $\cA/\cI$ has a Poincar\'e--Birkhoff--Witt basis $\cB$;
 \item $F$ is a Gr\"obner basis for $\cI$ with respect to the ordering $\prec$;
 \item For all $1\leqslant i<j<k\leqslant n$,
 \begin{equation}\label{pbw_conds}
      \Delta_{kji}=   (:\,(:s_k\,s_j:)\,  s_i:)-(:s_k\,(:s_j\,s_i:) :) =0.
       \end{equation}

\end{enumerate}
\end{thm}

In particular, the conditions (\ref{pbw_conds}) are necessary and sufficient for the linear independence of standard monomials in the quotient algebra $\cA/\cI$. These conditions ensure the uniqueness of normally ordered forms.
\medskip

We can use the above result to find necessary and sufficient conditions for the existence of a PBW basis of the quotient algebra $A/I$ in the case of ideals $I$ generated by polynomials (\ref{idJ0}). Let us consider the algebra $\cAA=\C\langle s_1,s_2,s_3\rangle$ and identify the variables $s_i$ with $x_i$ as
\begin{equation}\label{xtos3}
 s_1=x_1,\quad s_2=x_3,\quad s_3=x_2.
\end{equation}
Then the ideal $I\subset\cAA$ is generated by the polynomials
\begin{equation}
  \label{idJ0s}
 \begin{array}{l}
    f_{12}=s_2\,s_1-\omega
s_1\,s_2 ,\\
f_{13}=s_3\,s_1-\omega s_1\,s_3-s_1^2 \alpha _{11}-s_1\,s_2
\alpha _{14}-s_2^2 \alpha _{16}-s_1 \beta _{11}-s_2 \beta
_{13}-\gamma _1 ,\\
f_{23}= s_3\,s_2-\omega^{-1}s_2\,s_3+\omega^{-1}(s_1^2 \alpha
_{31}+s_1\,s_2 \alpha _{34}+s_2^2 \alpha _{36}+s_1 \beta
_{31}+s_2 \beta _{33}+\gamma
_3).
  \end{array}
 \end{equation}
Conditions (\ref{condF}) are clearly satisfied (see (\ref{exord})):
\[
s_1s_2 \prec s_2 s_1,\quad s_2^2\prec s_1s_3 \prec s_3 s_1 ,\quad s_2^2 \prec s_2s_3 \prec s_3 s_2.
\]
In the case $n=3$, there is only one equation (\ref{pbw_conds})
$$ \Delta_{321}=(:\,(:s_3\,s_2:)\,  s_1:)-(:s_3\,(:s_2\,s_1:) :) =0,$$
which must be satisfied. It follows from
\[
 \begin{array}{lcl}
  (:\,(:s_3\,s_2:)\,s_1:)&=&s_1\,s_2 \left(\beta _{11}-\beta _{33}\right)+s_1^2\,s_2 \left(\omega  \alpha _{11}-\alpha _{34}\right)+{s_2^2 \beta _{13}}{\omega^{-1} }+s_1\,s_2^2 \left(\alpha _{14}-\omega  \alpha _{36}\right)\\&+&{s_2^3 \alpha _{16}}{\omega^{-1} }-{s_1^3 \alpha _{31}}{\omega^{-1} }-{s_1^2 \beta _{31}}{\omega^{-1} }+{\gamma _1 s_2}{\omega^{-1} }-{\gamma _3 s_1}{\omega^{-1} }+s_1\,s_2\,s_3 \omega,\\&&\\
  (:s_3\,(:s_2\,s_1:) :)&=&s_1\,s_2 \left(\omega  \beta _{11}-\omega  \beta _{33}\right)+s_1^2s_2 \left(\omega  \alpha _{11}-\omega  \alpha _{34}\right)+s_2^2 \omega  \beta _{13}+s_1\,s_2^2 \left(\omega  \alpha _{14}-\omega  \alpha _{36}\right)\\&+&s_2^3 \omega  \alpha _{16}-s_1^3 \omega  \alpha _{31}-s_1^2 \omega  \beta _{31}+\gamma _1 s_2 \omega -\gamma _3 s_1 \omega +s_1\,s_2\,s_3 \omega,
 \end{array}
\]
that
\[
 \begin{array}{lcl}
 \Delta_{321} &=&(\omega-1)\Big(
s_1 s_2 \left(\beta _{33}-\beta _{11}\right)-s_1 s_2^2 \alpha _{14}+s_1^2s_2 \alpha _{34}\Big)\\
&+&\omega^{-1}(1-\omega^2)\Big(s_2^2 \beta _{13}+s_2^3 \alpha _{16}-s_1^3 \alpha _{31}-s_1^2 \beta _{31}+\gamma _1 s_2-\gamma _3 s_1\Big)=0.
\end{array}
\]
The above leads to three cases as stated in the theorem:
\begin{enumerate}
 \item $\omega=1$, with no conditions on the other coefficients;
 \item $\omega=-1$, with the conditions $\beta_{33}=\beta_{11}$ and $\alpha_{14}=\alpha_{34}=0$;
 \item $\beta_{33}=\beta_{11}$ and $\alpha_{14}=\alpha_{34}=\beta _{13}=\alpha _{16}=\alpha _{31}= \beta _{31}=\gamma _1=\gamma _3=0$. \hfill $\Box$
\end{enumerate}

\begin{rem}\label{remord3}
 The identification (\ref{xtos3}) defines the monomial ordering in $A$, such that the ideal $I$ in Lemma \ref{lem1} is generated by polynomials satisfying to the conditions (\ref{setF}), (\ref{condF}). Moreover, with this ordering, the leading term of any element $a\in A$ is $\T$--stable as  $\lm (\T(a))=\lm(a)$. The algebra $A$ with the ideal $I$ also admits a different monomial ordering, corresponding to the identification
 \begin{equation}\label{xtos3b}
  \hat{s}_1=x_3,\qquad \hat{s}_2=x_1,\qquad \hat{s}_3=x_2,
 \end{equation}
which similarly satisfies the above properties. This alternative monomial ordering corresponds to the involution introduced in Section \ref{inv} which correspond to reflecting the matrix in \eqref{deco} across its antidiagonal. \hfill $\Box$
\end{rem}

%
%======================= not yet done ================
%
%Then we extend this ordering to 6-dim case, find a toric ideal, formulate the general triangular ideal, state that an invariant ideal has a special form, then apply Levandovskii Theorem  and state the main result about quantisation ideals (main result of the section)

%\subsection{Квантование алгебры для 4-куба}

\subsection{The case $N(4,\cA)$, partial classification of quantisation ideals}
\label{seq-class}
Our main task is to describe the quantisation ideals of the free associative algebra $\cA = \mathbb{C} \left< x_1, \ldots, x_6 \right>$ such that the maps $\T_{ijk}$ defined as $\T$ on the corresponding components
\begin{equation}\label{extension}
\T_{ijk}(x_l) = \begin{cases}
x_j + x_i x_k, l = j, \\
x_l, l \ne j.
\end{cases}
\end{equation}
act as automorphisms.
In fact, we will consider only those maps that appear in the Zamolodchikov equation 
\begin{equation}\label{tetrahedron}
T _{123} \circ T _{145} \circ \T_{246} \circ \T_{356} = \T_{356} \circ \T_{246} \circ T _{145} \circ T _{123}.
\end{equation}
Due to the fact that $\T_{ijk}$ are homomorphisms of the free associative algebra by definition, we conclude that it is  {sufficient} to check the   identity \eqref{tetrahedron} only on the generators of the algebra $x_1, \ldots, x_6$. Checking the equality \eqref{tetrahedron} turns out to be non-trivial only for $x_4$
$$\LHS(x_4) = \RHS(x_4) = x_4 + x_1 x_5 + x_1 x_3 x_6 + x_2 x_6.$$
 
We adapt the notation from the previous chapter, namely, we introduce generators of the tensor algebra $\mathcal{A} = \mathbb{C} \left< s_1, \ldots, s_6 \right>$ with the monomial ordering described above \eqref{exord} due to the identification
\begin{equation}
s_1 = x_1, s_2 = x_3, s_3 = x_2, s_4 = x_6, s_5 = x_5, s_6 = x_4.
\end{equation}

Moving to variables $s_i$, it is convenient to change notations for the maps $\T_{ijk}$ in order to make them consistent with the indices: $\Phi_{132} = \T_{123}, \Phi_{165} = \T_{145}, \Phi_{364} = \T_{246}, \Phi_{254} = \T_{356}$. For example
\bea
\Phi_{132}(s_3) &=& s_3 + s_1 s_2=\T_{123}(x_2) = x_2 + x_1 x_3,\nn\\
\Phi_{165}(s_6) &=& s_6 + s_1 s_5,\ \ \, \Phi_{364}(s_6) = s_6 + s_3 s_4, \nn\\
\Phi_{254}(s_5) &=& s_5 + s_2 s_4.\nn
\eea 
Thus
\begin{equation}\label{phis}
\Phi_{ijk}(s_l) = \begin{cases}
s_j + s_i s_k, l = j, \\
s_l, l \ne j.
\end{cases}
\end{equation}
We introduce the set of admissible triples of indices $S = \{(1, 3, 2), (1, 6, 5), (3, 6, 4), (2, 5, 4)\}$ which correspond to the homomorphisms $\Phi_{ijk}$ that occur in the tetrahedron equation.

\iffalse
The ordering on the algebra generators agrees with the action of the homomorphisms $\Phi_{ijk}$, namely, they preserve the leading monomial
\[
 \lm(a)=\lm(\Phi_{ijk}(a)),\qquad \forall a\in\cA,\ (i,j,k)\in S\,.
\]
\fi
Recall that $\mathcal{M}$ denotes the set of all monomials in the free algebra $\mathcal{A}$, and that $\mathcal{B}$ refers to the set of standard monomials defined in (\ref{setB}). The map $\lm: \mathcal{A} \to \mathcal{M}$ assigns to each element of $\mathcal{A}$ its leading monomial with respect to the monomial ordering $\prec$ introduced in Definition \ref{monord}.
\begin{lemma}
For any $a \in \mathcal{M}$ and any admissible triple $(i, j, k) \in S$, $\lm(\Phi_{ijk}(a)) = a$
\end{lemma}
\begin{proof}
It is true that $\lm(a b) = \lm(a) \lm(b)$, $\Phi_{ijk}(a b) = \Phi_{ijk}(a) \Phi_{ijk}(b)$, so it suffices to prove the statement for $a = s_1, \ldots, s_6$. Note that $s_i s_k \prec s_j$, since $k < j$ for $(i,j,k) \in S$, then
\begin{equation*} 
\lm(\Phi_{ijk}(s_l)) = \lm(s_l + \delta_{l,j} s_i s_k) = s_l, 
\end{equation*}
where $\delta_{l,j}$ is the Kronecker delta.
\end{proof}

%\bea
%\label{toric_ideal}
%s_i s_j=\omega_{ij} s_j s_i.
%\eea
We begin the classification by studying ideals of {\em toric type}, denoted by $\cI_{g}$, which are generated by relations of the form
\bea
\label{toric_ideal}
g_{ij} = s_j s_i - \omega_{ij} s_i s_j, \quad \omega_{ji} = \omega_{ij}^{-1},
\eea
and are {\em  $\Phi$-stable}, that is, invariant under the automorphisms $\Phi_{i,j,k}$ for all $(i,j,k) \in S$.

Note that $\cI_{g}$ automatically satisfies the PBW condition; thus, we only need to verify its stability.
Theorem~\ref{thm-2-1} implies that several of the parameters must coincide:
\begin{equation}
\label{toric_cond}
 \begin{array}{ll}
\Phi_{132}:\ \ \omega_1 := \omega_{12} = \omega_{13} = \omega_{32}, \qquad & \Phi_{165}:\ \ \omega_2 := \omega_{15} = \omega_{16} = \omega_{65},\\
\Phi_{364}:\ \  \omega_3 := \omega_{34} = \omega_{36} = \omega_{64}, \qquad & \Phi_{254}:\ \ \omega_4 := \omega_{24} = \omega_{25} = \omega_{54}.
 \end{array}
\end{equation}
We introduce the shorthand notation $\omega_1, \ldots, \omega_4$ for simplicity. These parameters correspond to the automorphisms $\Phi_{i,j,k},\ (i,j,k) \in S$. The remaining parameters $\omega_{14}, \omega_{26}, \omega_{35}$ correspond to index pairs that do not simultaneously appear in any triple from $S$.

\begin{lemma}
\label{lem_toric_ideal}
A  $\Phi$-stable ideal of the form \eqref{toric_ideal} is uniquely determined by four parameters $\omega_1, \omega_2, \omega_3, \omega_4$ satisfying the relation
\bea\label{om14}
\omega_1 \omega_3 = \omega_2 \omega_4.
\eea
In this case, the parameters associated with the triplets in the tetrahedron equation are determined by the values of $\omega_i$ as specified in   \eqref{toric_cond}, while the remaining parameters are given by:
\bea\label{toric_cond2}
\omega_{14} = \omega_2 / \omega_1, \qquad \omega_{26} = \omega_3 / \omega_2, \qquad \omega_{35} = \omega_2 \omega_4.
\eea
\end{lemma}
\begin{proof}
Let $(i,j,k) \in S$ be any admissible triple, and let $r\notin \{i,j,k\}.$
%Введем также обозначения для генераторов идеала
%\bea
%g_{ij}=s_j s_i-\omega_{ij} s_i s_j.\nn
%\eea
We apply $\Phi_{ijk}$ to the generators of the ideal \eqref{toric_ideal}:
\bea
\Phi_{ijk}(g_{jr})&=&(\omega_{ir}\omega_{kr}-\omega_{jr})s_i s_k s_r{\color{blue}+g_{jr}+g_{ir}s_k+\omega_{ir}s_i g_{kr};}\nn\\
\Phi_{ijk}(g_{ij})&=&(\omega_{ik}-\omega_{ij})s_i^2 s_k{\color{blue}+g_{ij}+s_i g_{ik};}\nn\\
\Phi_{ijk}(g_{kj})&=&s_i s_k^2(1-\omega_{kj} \omega_{ik}){\color{blue}+g_{kj}-\omega_{kj} g_{ik} s_k;}
\label{phi_toric}
\eea
(the terms from the ideal are marked in blue).
Using the PBW property of this ideal, we obtain, in addition to the conditions formulated in Theorem \ref{thm-2-1}, the following set of equations:
\[
\begin{array}{llll}
\omega _{12} \omega _{14}=\omega _{15}, &\omega _{13} \omega _{14}=\omega _{16}, & \omega _{24}=\omega _{26}\omega _{32},&
\omega _{25}=\omega _{12} \omega _{26},\\
\omega _{14} \omega _{24}=\omega _{34},&
\omega _{15} \omega _{25}=\omega _{35}, &\omega _{34}\omega _{32}= \omega _{35}, &
\omega _{16} \omega _{26}=\omega _{36}, \\
\omega _{35}=\omega _{13} \omega _{36}, &\omega _{64}=\omega _{1,4} \omega _{54},&
\omega _{26} \omega _{65}=\omega _{64}, &\omega _{54} \omega _{65}=\omega _{35} .
\end{array}
\]
These 12 equations can be written uniformly as
\[
 \omega_{ni}\omega_{nk}=\omega_{nj},\qquad (i,j,k)\in S,\ \ n\not \in (i,j,k),
\]
and represented by the diagram in the figure

%\ref{pic_diag}.
%\begin{figure}[h!]
%\center
\[
\includegraphics[scale=.3]{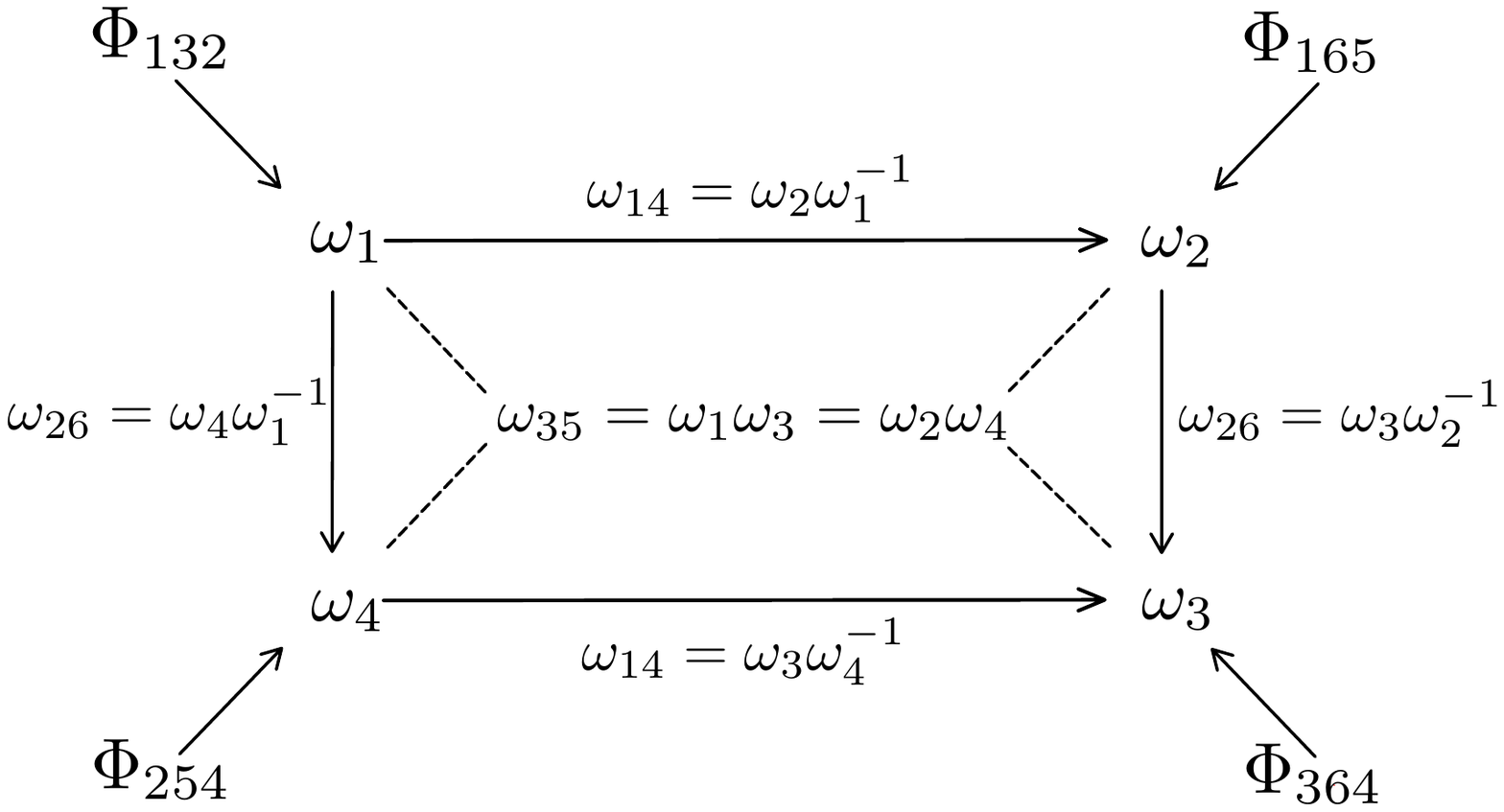}
\]%\includegraphics[]{toric}
%\includegraphics[]{toric_not}
%\caption{}
%\label{pic_diag}
%\end{figure}
%The legend for the arrows and dashed lines of the diagram is given by the %figure \ref{pic_toric_not}.

%\begin{figure}[h!]
%\center
%\includegraphics[]{toric_not}
%\caption{}
%\label{pic_toric_not}
%\end{figure}
The commutativity of the diagram guarantees the possibility of restoring all the parameters of the ideal from four $\omega_1,\omega_2,\omega_3,\omega_4$ if the following is satisfied
\bea
\omega_1 \omega_3=\omega_2 \omega_4.\nn
\eea
\end{proof}
Further on we will refer to invariant toric ideal as $I_0$ given by parameters $\omega_1, \omega_2, \omega_3, \omega_4$ subject to the relation $\omega_1 \omega_3 = \omega_2 \omega_4$.

We call a homogeneous quadratic ideal $\cI \subset \cA$ triangular if it is generated by polynomials of the form
\begin{equation}\label{ideal_type} 
f_{ij} := g_{ij} - d_{ij}; 1 \leqslant i < j \leqslant 6,
\end{equation}
where
\[
g_{ij}  = s_j s_i - \omega_{ij} s_i s_j, \qquad  
d_{ij}  = \sum\limits_{(k, l) < (i, j),\, k\leqslant l} \delta_{ijkl} s_k s_l.
\]
Here, $(k, l) < (i, j)$ denotes comparison in the right lexicographic order. This means that all monomials appearing in the quadratic polynomial $d_{ij}$ are standard, and that $\lm(d_{ij}) \prec s_i s_j$.

The toric component of an arbitrary triangular $\Phi$-stable ideal satisfy the same equations as listed in Lemma \ref{lem_toric_ideal}.
\begin{lemma}\label{lem_triangle}
Let a triangular ideal $\cI$ be $\Phi$-stable and  PBW, then parameters $\omega_{ij}$ satisfy the relations (\ref{toric_cond}), (\ref{toric_cond2}) and (\ref{om14}).
\end{lemma}

\begin{proof}
%Для удобства введем обозначения для отображений $\T_{ijk}$ в переменных $s_i$: $\Phi_{132} = \T_{123}, \Phi_{165} = \T_{145}, \Phi_{364} = \T_{246}, \Phi_{254} = \T_{356}$, а именно $\Phi_{132}(s_3) = \T_{123}(x_2) = x_2 + x_1 x_3 = s_3 + s_1 s_2, \Phi_{165}(s_6) = s_6 + s_1 s_5, \Phi_{364}(s_6) = s_6 + s_3 s_4, \Phi_{254}(s_5) = s_5 + s_2 s_4$, то есть $\Phi_{ijk}(s_j) = s_j + s_i s_k$. Во всех четырех случаях $\Phi_{ijk}$ выполняется  $i < k < j$. Также введем произвольные индексы $l, r$, такие что $l < j < r$.

Note that due to the homogeneity of the ideal $\cI$, the quotient algebra $\cA / \cI$ is graded. In addition, due to the PBW property, we have an isomorphism of graded vector spaces
\bea
\tau:\mathbb{C}[s_1, \ldots, s_6]\rightarrow\cA / \cI ,\nn
\eea
which is determined by the choice of the basis of ordered monomials.
Consider the projection operator onto the degree three component
\bea
P_{3,\cA/\cI}: \cA / \cI \to (\cA / \cI)_3.\nn
\eea
This operator is consistent with the projection operator $P_{3,\cA}$ in the free algebra $\cA$
\bea
\pi P_{3,\cA}= P_{3,\cA/\cI} \pi,\nn
\eea
where $\pi: \cA \to \cA / \cI$ is the canonical projection. From now on we will use the notation $P_3$ for both operators. Finally, we introduce the operator $\mathcal{P}_3 = \tau^{-1} \circ P_3 \circ \pi: \mathcal{A} \to (\mathbb{C}[s_1, \ldots, s_6])_3$, which we will work with, and $\Pi = \tau^{-1} \circ \pi: \mathcal{A} \to \mathbb{C}[s_1, \ldots, s_6]$. Recall that $\pi_0: \mathcal{M} \to \mathcal{B}$ is a map that, given an arbitrary monomial, constructs a standart monomial equivalent to it (in the sense of the introduced equivalence relation).
 % -- проекцию на подпространство, которое порождено образами мономов третьей степени из $\mathbb{C}[s_1, \ldots, s_6] \subset \mathbb{C} \left< s_1, \ldots, s_6 \right> = \cA$. Градуировка корректно определена на факторалгебре, потому что идеал однородный, а упорядоченные мономы являются линейным базисом из-за выполнения теоремы ПБВ.
Thanks to the relations (\ref{phi_toric}) we can say that in the toric case
\bea
\mathcal{P}_3(\Phi_{ijk}(g_{ij}))\in\C\pi_0(s_i^2 s_k);\qquad \mathcal{P}_3(\Phi_{ijk}(g_{jr}))\in\C\pi_0(s_i s_k s_r);\qquad \mathcal{P}_3(\Phi_{ijk}(g_{kj}))\in\C\pi_0(s_i s_k^2).\nn
\eea
The relations of Lemma \ref{lem_toric_ideal} guarantee the invariance of the ideal in the toric case. Now consider the invariance conditions for the general triangular ideal.

First of all, we note that the notion of a leading monomial on a free algebra for ideals of the type under consideration can be extended to the quotient algebra and coincides with the standard monomial $\lm(\Pi(a)))=\pi_0(a)$ for any monomial $a \in \mathcal{M}$ if we consider $\mathbb{C}[s_1, \ldots, s_6] = \Span_{\mathbb{C}} \mathcal{B} \subset \mathcal{A}$ as a vector space. In other words, the following diagram is commutative.

\[\begin{tikzcd}
\mathcal{M} \arrow{r}{\pi_0} \arrow[swap]{d}{\pi} & \mathcal{B} \\
\mathcal{A}/\mathcal{J} \arrow{r}{\tau^{-1}} & \Span_{\mathbb{C}} \mathcal{B} \arrow{u}{\lm}
\end{tikzcd}
\]

We prove that $\mathcal{P}_3(\Phi_{ijk}(d_{lj}))$ for $l<j$ does not contain the monomial $\pi_0(s_i s_k s_l)$. This will mean that the condition on the coefficient of $\pi_0(s_i s_k s_l)$ for a general ideal coincides with the condition for a toric one. To do this, we show that $\lm(P_3(\Phi_{ijk}(d_{lj}))) \prec \pi_0 (s_i s_k s_l)$. In fact, let's calculate the cubic part
\begin{equation*} 
\mathcal{P}_3(\Phi_{ijk}(\sum_{(p,q) < (l,j)}\delta_{ljpq} s_p s_q)) = \mathcal{P}_3(\sum_{p < l} \delta_{ljpj} s_p s_i s_k) = \mathcal{P}_3(\sum_{p < l} \delta_{ljpj} \pi_0(s_p s_i s_k))
\end{equation*}
It is enough to restrict the order to standard monomials of the same degree, that is,  to consider the right (reverse) lexicographical ordering. 
We get
\begin{equation*} 
\lm(\mathcal{P}_3(\Phi_{ijk}(d_{lj}))) 
%= \pi_0(s_{l-1} s_i s_k) 
\prec \pi_0(s_i s_k s_l)
\end{equation*}
Let us carry out a similar argument to check that $\mathcal{P}_3(\Phi_{ijk}(d_{jr}))$ does not contain the monomial $\pi_0(s_i s_k s_r)$ for $r>j$.
\begin{equation*} 
\lm(\mathcal{P}_3(\Phi_{ijk}(d_{jr}))) 
% = \pi_0(s_i s_k s_{r-1}) = s_i s_k s_{r-1} 
\prec \pi_0(s_i s_k s_r)
\end{equation*}
The cases $\mathcal{P}_3(\Phi_{ijk}(d_{ij}))$ and $\mathcal{P}_3(\Phi_{ijk}(d_{jk}))$ are checked similarly.
The expressions $d_{l' r'}$ for $ l' \ne j$ and $r' \ne j$ we have $\mathcal{P}_3(\Phi_{ijk}(d_{l'r'})) = 0$, since $\Phi_{ijk}(d_{l'r'})$ does not contain cubic part.
\end{proof}

The following theorem provides a strong characterisation of triangular $\Phi$-stable ideals.
\begin{thm}\label{thm_general_pbw}
Let  $\cI$ be a triangular $\Phi$-stable PBW ideal, then $d_{ij} = 0$ for $j < 6$.
\end{thm}
\begin{proof} {(\it Sketch}) 
Assume that the ideal $\cI$ is $\Phi$-stable and satisfies the PBW condition. Then for each $(i,j,k) \in S$, we must have 
\[
\mathcal{P}_3(\Phi_{ijk}(f_{sp})) \equiv 0
\]
The idea of the proof is to successively establish that each term $d_{ij}$ vanishes for $j < 6$. The PBW condition, together with $\Phi$-invariance, imposes constraints on the coefficients $\delta_{ijkl}$ in the expressions for $d_{ij}$. At each step, we derive these constraints and observe that they reduce to trivial equations of the form $\delta_{ijkl} = 0$, which can be solved recursively. This procedure eventually leads to the conclusion that all such correction terms must vanish. \\
\textit{Step 1}.\\\\
By Lemma~\ref{lem_triangle}, the parameters $\omega_{ij}$ satisfy the necessary consistency conditions. Consider now the action of $\Phi_{132}$ on the generator $f_{13}$. We compute:
\[
\mathcal{P}_3\left(\Phi_{132}(f_{13})\right) = s_1 s_2 s_1 - \omega_{13} s_1^2 s_2 = (\omega_{12} - \omega_{13}) s_1^2 s_2 + \delta_{1211} s_1^3.
\]
Since the PBW condition requires $\mathcal{P}_3\left(\Phi_{132}(f_{13})\right) \equiv 0$, and by Lemma~\ref{lem_triangle} we already know $\omega_{12} = \omega_{13}$, we conclude:
\[
\delta_{1211} s_1^3 = 0 \quad \Rightarrow \quad \delta_{1211} = 0.
\]
Therefore, the  term $d_{12}$ vanishes, and the relation for $s_2 s_1$ is purely of toric type.

This, in turn, implies that any terms involving $s_3$ appearing in $d_{ij}$ for $i,j \neq 3$ must vanish, i.e.,
\[
\delta_{ijk3} = \delta_{ij3k} = 0 \quad \text{for all } i, j \neq 3.
\]
 Indeed:
    \begin{align*}
        &\mathcal{P}_3(\Phi_{132} (f_{ij})) = \mathcal{P}_3(\Phi_{132}(\sum_{\substack{p < 3 \\ (p, 3) < (i, j)}} \delta_{ikp3} s_p s_3)) + \mathcal{P}_3(\Phi_{132}(\sum_{\substack{p > 3 \\ (3, p) < (i, j)}} \delta_{ij3p} s_3 s_p)) = \\ &= \sum_{\substack{p < 3 \\ (p, 3) < (i, j)}} \delta_{ikp3} \pi(s_p s_1 s_2) + \sum_{\substack{p > 3 \\ (3, p) < (i,j)}} \delta_{ij3p} s_1 s_2 s_p = \\ &=
        \sum_{\substack{p < 3 \\ (p, 3) < (i, j)}} \delta_{ijp3} \omega_{1p} \pi_0(s_p s_1 s_2) + \sum_{\substack{p > 3 \\ (p, 3) < (i,j)}} \delta_{ij3p} s_1 s_2 s_p = 0 \implies \delta_{ij3p} = \delta_{ijp3} = 0
    \end{align*}
 The projection $\mathcal{P}_4$ onto monomials of degree 4 yields another relation, namely $\delta_{ij33} = 0$. This follows because the only monomial of degree 4 appearing in the expression is
\[
\delta_{ij33} \cdot \Pi(s_1 s_2 s_1 s_2) = \delta_{ij33} \omega_{12} s_1^2 s_2^2.
\]
Similarly, applying $\mathcal{P}_4$ to $\Phi_{254}(f_{ij})$ gives the relation
\[
\delta_{ij55} = 0.
\]
\\
\textit{Step 2}.
\\
It is not feasible to eliminate such a large number of coefficients through general reasoning alone. However, we will demonstrate the principle by which we proceed. As an illustration, we will show how to eliminate $d_{14}$ in two steps.
We have:
\begin{equation*}
    \mathcal{P}_3(\Phi_{132}(f_{23})) = \Pi(s_1 s^2_2 - \omega_{23} s_2 s_1 s_2 - \delta_{2313} s_1 s_1 s_2) = (1 - \omega_{23} \omega_{12}) s_1 s^2_2 - \delta_{2313}s^2_1 s_2 \implies \delta_{2313} = 0  
\end{equation*}
\begin{align*}
    \mathcal{P}_3(\Phi_{132}(f_{34})) = \Pi(s_4 s_1 s_2 - \omega_{34} s_1 s_2 s_4 - \delta_{3413} s^2_1 s_2 - \delta_{3423} s_2 s_1 s_2) = \\ = \Pi(s_4 s_1 s_2) - \omega_{34} s_1 s_2 s_4 - \delta_{3413} s^2_1 s_2 - \delta_{3423} \omega_{12} s_1 s^2_2\\
    \Pi(s_4 s_1 s_2) = \omega_{14} \omega_{24} s_1 s_2 s_4 + (\delta_{1411} + \omega_{14} \delta_{2412}) s^2_1 s_2 + (\delta_{1412} + \omega_{14} \delta_{2422}) s_1 s^2_2 + \\ + \omega_{14} \delta_{2411} s^3_1 + \delta_{1422} s^3_2 + \omega_{14} \delta_{2414} s^2_1 s_4
\end{align*}
From the previous calculations, we observe that several coefficients in $d_{14}$ and $d_{24}$ vanish immediately, namely $\delta_{1422} = \delta_{2414} = \delta_{2411} = 0$.

Next, a similar direct computation yields
\begin{equation*}
\mathcal{P}_3(\Phi_{254}(f_{15})) = \omega^2_{12} \delta_{1411} s^2_1 s_2 + \omega_{12} \delta_{1412} s_1 s^2_2 \implies \delta_{1412} = \delta_{1411} = 0
\end{equation*}
As a consequence, we obtain $d_{14} = 0$.
\\
\textit{Step 3}.
\\
Next, we show that $d_{24} = 0$. We compute
\[
\mathcal{P}_3(\Phi_{254}(f_{25})) = \omega_{12} \delta_{2412} s_1 s^2_2 + \delta_{2422} s^3_2 - \delta_{2515} s_1 s_2 s_4.
\]
We conclude that $\delta_{2412} = \delta_{2422} = \delta_{2515} = 0$, and hence $d_{24} = 0$.

Returning to the previous step, we also find that some of the remaining coefficient equations are now trivially satisfied, namely $\delta_{3423} = \delta_{3413} = 0$.
\\
\textit{Remaining steps}.
\\
We now indicate which automorphisms $\Phi_{ijk}$ should be applied to which generators $f_{kl}$ in order to deduce the vanishing of the remaining $d_{kl}$. As before, we resolve only trivial equations of the form $\delta_{ijkl} = 0$ (using that $\omega_{ij} \ne 0$), and we also revisit previous coefficient equations since simplifications may occur.
\begin{align*}
    \Phi_{254}(f_{35}), \Phi_{123}(f_{35}), \Phi_{165}(f_{16}) \implies d_{15} = 0 \\
    \Phi_{165}(f_{26}) \implies d_{25} = 0 \\
    \Phi_{165}(f_{36}), \Phi_{364}(f_{16}) \implies d_{13} = 0 \\
    \Phi_{364}(f_{26}) \implies d_{23} = 0 \\
    \Phi_{364}(f_{46}) \implies d_{34} = d_{35} = 0 \\
    \Phi_{132}(f_{45}), \Phi_{165}(f_{45}), \Phi_{165}(f_{56}) \implies d_{45} = 0
\end{align*}
\end{proof}

Theorem \ref{thm_general_pbw} significantly simplifies the classification of triangular PBW ideals that are invariant under the automorphisms $\Phi_{ijk}$. It enables us to find the generators $f_{i,j}$ for the most general  triangular ideal $I_{stab}$, that is, an ideal satisfying stability under all maps $\Phi_{ijk}$ but not necessarily the PBW property. The generators of the ideal are given by:
\begin{equation}\label{J-stab}
    \begin{cases}
f_{12} = s_2 s_1 - \omega_1 s_1 s_2, \\ f_{13} = s_3s_1 - \omega_1 s_1 s_3, \\ f_{23} = s_3s_2 - \omega_1^{-1} s_2 s_3,\\

f_{14} = s_4s_1 - \omega_2 \omega_1^{-1} s_1 s_4, \\ f_{24} = s_4 s_2 - \omega_4 s_2 s_4, \\ f_{34} = s_4s_3 - \omega_3 s_3 s_4,\\

f_{15} = s_5s_1 - \omega_2 s_1 s_5, \\ f_{25} = s_5s_2 - \omega_4 s_2 s_5, \\ f_{35} = s_5s_3 - \omega_2 \omega_4 s_3 s_5, \\ f_{45} = s_5s_4 - \omega_4^{-1} s_4 s_5,\\

f_{16} = s_6 s_1 - \left(\omega_2 s_1 s_6 + A_3 s_1^2 + A_1 \omega_1 s_1 s_2 + A_2 s_1 s_4\right),\\

f_{26} = s_6 s_2 - \left(\omega_3 \omega^{-1}_2 s_2 s_6  + A_4 s_1 s_2 + A_6 \omega_1 s_2^2 + A_5 \omega_4 s_2 s_4\right),\\

f_{36} = s_6 s_3 - \left( \omega_3 s_3 s_6 + A_{10} s_1^2 + A_7 s_1 s_2 + \left(A_4 \omega_2 + A_3\right) s_1 s_3  + A_8 s_1 s_4\right) - \\ \qquad - \left(A_{11} s_2^2 + s_2 s_3 \left(A_6 \omega_2+A_1\right) + A_9 s_2 s_4 + \omega_4 s_3 s_4 \left(A_5 \omega_2+A_2\right) + A_{12} s_4^2\right),\\

f_{46} = s_6 s_4 - \left(\omega^{-1}_3 s_4 s_6 + A_{13} s_1 s_4 + A_{14} \omega_2 s_2 s_4 + A_{15} \omega_1 s_4^2\right),\\

f_{56} = s_6 s_5 - \left(\omega^{-1}_2 s_5 s_6 + A_{19} s_1^2 + A_{16} s_1 s_2 + A_{17} s_1 s_4 + \left(A_{13} \omega_4+A_4\right) s_1 s_5 \right) - \\ \qquad - \left(A_{20} s_2^2 + A_{18} s_2 s_4 + \left(A_6 \omega_1+A_{14} \omega_3\right) s_2 s_5 + A_{21} s_4^2 + \left(A_5+A_{15}\right) s_4. s_5\right)
\end{cases}
\end{equation}
In this system, all parameters $\omega_i$ and $A_i$ are arbitrary, except for the constraint  $\omega_1 \omega_3 = \omega_2 \omega_4$.
\medskip

It follows from Theorem \ref{LevThm} that the ideal $I_{stab}$ is PBW if and only if
\begin{equation}\label{pbw1}
       \big(:\,(:s_k\,s_j:)\,  s_i:\big) = \big(:s_k\,(:s_j\,s_i:):\big), \qquad 1\leqslant i<j<k\leqslant 6.
\end{equation}
Conditions \eqref{pbw1} yield a system of 78 linear homogeneous equations for the unknowns \(A_1,\ldots,A_{21}\), with coefficients in the ring \(\mathbb{Z}[\omega_i,\omega_i^{-1}\,;\, i=1,\ldots,4]/(\omega_1\omega_3-\omega_2\omega_4)\). The solvability conditions for this system impose polynomial constraints on the toric parameters \(\omega_i\). The solution set defines an affine variety with several irreducible components. A complete classification of all solutions is challenging and currently beyond our reach.

In this paper, we focus on {\em the generic case}, where none of the parameters \(\omega_i \) is identically equal to $\pm 1$.
The generic case can be reduced to the following four subcases:
\bea\label{c1}
\mbox{\bf Case 1} &:\qquad&   \omega_1 \omega_3 = \omega_2 \omega_4,\\ \label{c2}
\mbox{\bf Case 2} &:\qquad& \omega_1 = \omega_2 = \omega_3 = \omega_4,   \\
\label{c3}
\mbox{\bf Case 3} &:\qquad& \omega_1 = \omega_2,\quad \omega_3 = \omega_4 = \omega_2^2, \\
\label{c4}
\mbox{\bf Case 4} &:\qquad& \omega_1 = \omega_2 = \omega_4^2,\quad  \omega_3 = \omega_4.
\eea
Below we present generating polynomials for the most general $\Phi$-stable and PBW ideals corresponding to each Case.For brevity, we show only the polynomials $f_{i,j},\ j=6$; the generators  $f_{i,j}$ with $j<6$ are all toric and given by \eqref{J-stab} specialized to each case. The arbitrary constants $b_i, C_i$ appearing here are related to, but do not coincide with, the constants $A_i$ in \eqref{J-stab}.
\medskip

\noindent
{\bf Toric ideal $I_0$:}\\ 
Let $I_0$ denote a $\Phi$-stable ideal generated by the polynomials $g_{i,j}$ defined in \eqref{toric_ideal}, with the parameters $\omega_{i,j}$ specified in Lemma~\ref{lem_toric_ideal}:
\begin{equation}\label{I0}
    \left\{ \begin{array}{lll}
g_{12} = s_2 s_1 - \omega_1 s_1 s_2, & 
g_{13} = s_3s_1 - \omega_1 s_1 s_3, & 
g_{23} = s_3s_2 - \omega_1^{-1} s_2 s_3,\\
g_{14} = s_4s_1 - \omega_2 \omega_1^{-1} s_1 s_4, & 
g_{24} = s_4 s_2 - \omega_1\omega_3\omega_2^{-1} s_2 s_4, & 
g_{34} = s_4s_3 - \omega_3 s_3 s_4,\\
g_{15} = s_5s_1 - \omega_2 s_1 s_5, & 
g_{25} = s_5s_2 - \omega_1\omega_3\omega_2^{-1} s_2 s_5, & 
g_{35} = s_5s_3 - \omega_1 \omega_3 s_3 s_5, \\ 
g_{45} = s_5s_4 - \omega_2\omega_1^{-1}\omega_3^{-1} s_4 s_5,&
g_{16} = s_6 s_1 -  \omega_2 s_1 s_6  ,&
g_{26} = s_6 s_2 -  \omega_3 \omega^{-1}_2 s_2 s_6  ,\\
g_{36} = s_6 s_3 -   \omega_3 s_3 s_6 ,&
g_{46} = s_6 s_4 - \omega^{-1}_3 s_4 s_6  ,&
g_{56} = s_6 s_5 -  \omega^{-1}_2 s_5 s_6  .
    \end{array}\right.
\end{equation}
This ideal is parametrised by three independent non-zero parameters $\omega_1, \omega_2$ and $\omega_3$.
\medskip

\noindent
{\bf Ideal $I_1$:}\\
{\bf Case 1} yields the ideal $I_1$,  generated by polynomials $f_{i,j}=g_{i,j}, j<6$ (\ref{I0}), and:
\begin{equation}
    \begin{cases}
f_{16} = s_6 s_1 - \left(\omega_2 s_1 s_6 + b_1 \left(1-\omega_2\right) s_1^2+b_2 \left(\omega_1-\omega_2\right) s_1 s_2 + b_3 (\omega^{-1}_1 - 1) \omega_2 \omega_3 s_1 s_4\right),\\f_{26} = s_6 s_2 - \left(\omega_3 \omega^{-1}_2 s_2 s_6 + b_1 (1 - \omega_1 \omega_3 \omega^{-1}_2) s_1 s_2 + b_2 (1 - \omega_3 \omega^{-1}_2) s^{2}_2 + b_3 (\omega_1 - 1) \omega^{2}_3 \omega^{-1}_2 s_2 s_4 \right),\\f_{36} = s_6 s_3 - \left(\omega_3 s_3 s_6 + b_1 \left(1 - \omega_1 \omega_3\right) s_1 s_3 + b_2 (1 - \omega_3 \omega^{-1}_1) s_2 s_3\right),\\f_{46} = s_6 s_4 - \left(\omega^{-1}_3 s_4 s_6 + b_1 (1 - \omega_2 \omega^{-1}_1 \omega^{-1}_3) s_1 s_4 + b_2 (1 - \omega_1 \omega^{-1}_2) s_2 s_4 + b_3 \left(\omega_3-1\right) s_4^2\right),\\f_{56} = s_6 s_5 - \left(\omega^{-1}_2 s_5 s_6 + b_2 (1 - \omega_1 \omega_3 \omega^{-2}_2) s_2 s_5 + b_3 (\omega_3 - \omega^{-1}_1) s_4 s_5\right),
    \end{cases}
\end{equation}
where $b_1,b_2,b_3$ and  $ \omega_1,\omega_2,\omega_3 $ are arbitrary parameters.

\begin{rem}
Ideal $I_1$ is a deformation of the generic toric ideal $I_0$, it coincides with  $I_0$ if $b_i = 0$.
\end{rem}

\noindent
{\bf Ideal $I_2$:}\\
{\bf Case 2} yields the ideal $I_2$, generated by polynomials $f_{i,j}=g_{i,j}, j<6$ (\ref{I0}), where $\omega_1=\omega_2=\omega_3=\omega$ and
\begin{equation}
 \begin{cases}
f_{16} = s_6s_1 - \left(\omega  s_1 s_6 -C_5 \omega ^2 s_1 s_4+C_1 \omega  s_1^2+\left(C_3-C_2\right) \omega  s_1 s_2\right),\\ f_{26} = s_6s_2 - \left(s_2 s_6 + C_5 \omega ^2 s_2 s_4+C_1 \omega  s_1 s_2+C_2 s_2^2\right),\\ f_{36} = s_6s_3 - \left(\omega  s_3 s_6 + C_1 \omega  \left(\omega +1\right) s_1 s_3+C_3 s_2 s_3\right),\\f_{46} = s_6s_4 - \left(\omega ^{-1} s_4 s_6 + C_5 \omega  s_4^2-C_1 s_1 s_4+C_4 s_2 s_4\right),\\f_{56} = s_6s_5 - \left(\omega ^{-1} s_5 s_6 + C_5 \left(\omega +1\right) s_4 s_5+\left(C_2+C_4\right) s_2 s_5\right),
    \end{cases}
\end{equation}
This ideal depends on arbitrary parameters $C_1,C_2,C_3,C_4,C_5$ and $\omega\ne 0$.

\noindent
{\bf Ideal $I_3$:}\\
{\bf Case 3} yields the ideal $I_3$, generated by polynomials $f_{i,j}=g_{i,j}, j<6$ (\ref{I0}), where $\omega_1=\omega_2=\omega,\ \omega_3=\omega^2$ and
\begin{equation}
     \begin{cases}
f_{16} = s_6s_1 - \left(\omega  s_1 s_6 + C_1 \omega ^2 s_1^2 + C_2 \omega ^3 s_1 s_4\right),\\

f_{26} = s_6s_2 - \left(\omega  s_2 s_6 + C_1 \left(\omega +1\right) \omega ^2 s_1 s_2 + C_4 s_2^2 -C_2 \omega ^4 s_2 s_4\right),\\

f_{36} = s_6s_3 - \left(\omega ^2 s_3 s_6 + C_1 \left(\omega ^2+\omega +1\right) \omega ^2 s_1 s_3+C_3 s_2^2+C_4 s_2 s_3\right),\\

f_{46} = s_6s_4 - \left(\omega ^{-2} s_4 s_6 - C_1 \left(\omega +1\right) s_1 s_4 - C_2 \left( \omega +1\right) \omega  s_4^2\right),\\

f_{56} = s_6s_5 - \left(\omega ^{-1} s_5 s_6 +C_4 s_2 s_5 -C_2 \left(\omega ^2+\omega +1\right) s_4 s_5\right),
    \end{cases}
\end{equation}
This ideal depends on arbitrary parameters $C_1,C_2,C_3,C_4$ and $\omega\ne 0$.

\noindent
{\bf Ideal $I_4$:}\\
{\bf Case 4} yields the ideal $I_4$,  generated by polynomials $f_{i,j}=g_{i,j}, j<6$ (\ref{I0}), where $\omega_1=\omega_2=\omega^2,\ \omega_3=\omega$ and
\begin{equation}
     \begin{cases}
f_{16} = s_6s_1 - \left(\omega ^2 s_1 s_6 - C_3 \left(\omega +1\right) \omega ^3 s_1 s_4+C_1 \left(\omega +1\right){}^2 \omega  s_1^2\right),\\

f_{26} = s_6s_2 - \left(\omega ^{-1} s_2 s_6 + C_1 \left(\omega +1\right) \omega  s_1 s_2+C_2 s_2^2 + C_3 \left(\omega +1\right) \omega ^2 s_2 s_4\right),\\

f_{36} = s_6s_3 - \left(\omega  s_3 s_6 + C_1 \omega  \left(\omega +1\right) \left(\omega ^2+\omega +1\right) s_1 s_3+C_2 s_2 s_3\right),\\

f_{46} = s_6s_4 - \left(\omega ^{-1} s_4 s_6 -C_1 \left(\omega +1\right) s_1 s_4 + C_3 \omega ^2 s_4^2\right),\\

f_{56} = s_6s_5 - \left(\omega ^{-2} s_5 s_6 + C_4 s_2^2+C_2 s_2 s_5 + C_3 \left(\omega ^2+\omega +1\right) s_4 s_5\right) ,   
\end{cases}
\end{equation}
This ideal depends on arbitrary parameters $C_1,C_2,C_3,C_4$ and $\omega\ne 0$.

%\red{
%\begin{rem}
%If an ideal $I\subset\cA=\C\langle x_1,\ldots,x_6\rangle$ is preserved under the automorphisms   $T_{ijk}$ of $\cA$, the tetrahedron equation is satisfied in the corresponding quotient algebra, then the same will be true in $\C\langle y_1,\ldots,y_6\rangle$ with   inverse multiplication.
%\end{rem}
%}

\begin{rem}
Similar to the case with three variables (see Remark~\ref{remord3}), there exists an alternative monomial ordering corresponding to the involution induced by reflecting the matrix in \eqref{M4} across its antidiagonal. This reflection gives rise to an involution of the generators:
 \[
  s_1 \leftrightarrow s_4,\qquad s_3\leftrightarrow s_5.
 \]
Under this involution, the ideals $I_0, I_1$ and $I_2$ are mapped to similar ideals with a different choice of parameters, while the ideals of type $I_3$ and $I_4$ are interchanged (with transformations of parameters).
\end{rem}

\begin{rem} In this paper, we do not study non-generic ideals corresponding to the case in which at least one of the parameters $\omega_1,\omega_2,\omega_3$ is identically equal to $\pm1$. Thus,  we exclude Cases 1 and 2 from Theorem~\ref{thm-2-1}.
Such non-generic ideals do exist. For example, a non-generic ideal $I_{*}$ is generated by the polynomials
\begin{equation}\label{other-ideal}
\left\{ 
\begin{array}{lll}
 f_{12} = s_2 s_1 - \omega  s_1 s_2,& f_{13} = s_3 s_1 - \omega  s_1 s_3,& f_{23} = s_3 s_2 - \omega^{-1} s_2 s_3,\\ f_{14} = s_4 s_1 - \omega^{ 2} s_1 s_4, &f_{24} = s_4 s_2 - \omega^{-2} s_2 s_4,& 
 f_{34} = s_4 s_3 - s_3 s_4,\\
 f_{15} = s_5 s_1 - \omega^{3} s_1 s_5,& f_{25} = s_5 s_2 - \omega^{-2} s_2 s_5,&
 f_{35} = s_5 s_3 - \omega  s_3 s_5, \\ 
 f_{45} = s_5 s_4 - \omega^{2} s_4 s_5,& f_{16} = s_6 s_1 - \omega^{3} s_1 s_6,&f_{26} = s_6 s_2 - \omega^{-3} s_2 s_6, \\ 
 f_{36} = s_6 s_3 -  s_3 s_6 - C s_4^2 , &f_{46} = s_6 s_4 - s_4 s_6, &f_{56} = s_6 s_5 - \omega^{-3} s_5 s_6,
\end{array}\right.
\end{equation}
This ideal corresponds to the parameter values
\[\omega_1=\omega,\quad \omega_2=\omega^{3},\quad \omega_3=1,\]
and does not arise as a specialisation of any of the generic ideals $I_1,\ldots, I_4$ if the arbitrary parameter $C\ne 0$.
\end{rem}

\begin{rem}
The sequence of re-parametrisations and corresponding  charts $C_i^{\alpha}$ from Section \ref{sec_param} are well defined on the quantum algebra $\mathcal{A}_I$.  This is due to stability of the idel $I$ with respect to the mutations $\Phi_{ijk},\ (i,j,k)\in S$.
\end{rem}

%Будем говорить, что одно семейство идеалов $I(\alpha)$ вкладывается в другое $J(\beta)$, если для любого набора параметров $\tilde{\alpha}$ существует набор $\tilde{\beta}$, такой что $I(\tilde{\alpha}) = J(\tilde{\beta})$. Если $I(\alpha)$ не вкладывается в $J(\beta)$ и наоборот, то будем говорить, что такие идеалы независимы. 
%\begin{lemma}
%    Идеалы $J_1, J_2, J_3$ и $J_4$ являются попарно независимыми
%\end{lemma}
%\begin{proof}
%    Дальнейшие рассуждения справедливы, потому что выписанные соотношения для этих идеалов образуют базис Грёбнера, а значит являются линейно независимыми.
    
%    Заметим, что соотношение на $s_6 s_3$ у $J_2$ содержит моном $s_2^2$ в отличии от $J_1$ и $J_2$, то есть пары $(J_1, J_2)$ и $(J_2, J_3)$ таким образом являются независимыми. Теперь посмотрим на $J_4$. Коэффициенты при мономе $s_2 s_3$ в соотношении для $s_6 s_3$ и при мономе $s_2 s_4$ в соотношении для $s_6 s_4$ являются зависимыми, то есть один можно алгебраически выразить через другой. Тем временем, в $J_1$ соответствующие коэффициенты независимы. Таким образом пара идеалов $(J_1, J_4)$ также независима.

%    Далее, идеал $J_3$ имеет ненулевой коэффицент при мономе $s_2^2$, в отличии от остальных идеалов, поэтому в него не могут вкладываться $J_1$, $J_2$ и $J_4$.
%\end{proof}

\section{Classical limit}

All quantum algebras obtained in this paper (with the exception of Case 2 in Theorem~\ref{thm-2-1}) can be viewed as deformations of commutative polynomial rings. It is well known (first observed by Dirac in 1925 \cite{Dirac25}) that the classical limit of commutators yields Poisson brackets (see, for example, \cite{PLV}). Deformations of noncommutative algebras also give rise to Poisson structures \cite{MV}. However, in this paper we restrict our attention to standard deformations of commutative algebras.

Before presenting the result for the case of the group $N(4,\cA/I)$ , we illustrate the classical limit for the group $N(3,\cAA/I)$ in the   Case 3 of Theorem~\ref{thm-2-1} in detail. The parameters of the ideal $I$  (\ref{case3}) we represent as
\[
 \omega=1+\nu a,\quad \alpha_{11}=\nu b,\quad \alpha_{36}=\nu c,\quad \beta_{33}=\nu d,
\]
where $\nu$ is the deformation parameter and $a,b,c,d\in \C$ are arbitrary parameters. In the case $\nu=0$ the quotient algebra is just a commutative polynomial ring $\cAA/(I|_{\nu=0})=\C[x_1,x_2,x_3]$. We further assume that $x_1$ and $x_3$ are invertible, which allows us to define the localised algebra:
\[\hat{\cAA}_I=\C\langle x_1,x_1^{-1},x_2,x_3,x_3^{-1} \rangle /I.\]   The centre of the algebra $\hat{\cAA}_I$ is generated by the element
\[
 z=x_1^{-1}(a x_2+b x_1 +c x_3+d)x_3^{-1}.
\]
The Poisson brackets for any two elements $a,b$ of $\C [x_1,x_2,x_3]$, corresponding to the classical limit $\nu\to 0$ are defined as
\[
 \{a,b\}=\lim\limits_{\nu\to 0}\frac1{\nu}[a,b]=\sum\limits_{i=1}^3\sum\limits_{j=1}^3  \{x_i,x_j\}\frac{\partial a}{\partial x_i}\frac{\partial b}{\partial x_j},\qquad \{x_i,x_j\}=\lim\limits_{\nu\to 0}\frac1{\nu}[x_i,x_j].
\]
In our case
\bea
 \nn
 \{x_2,x_1\}&=&\frac1{\nu}[x_2,x_1]=a x_2 x_1+b x_1^2+d x_1,\\  
 \{x_3,x_1\}&=&\frac1{\nu}[x_3,x_1]=a x_3 x_1,\\ \nn
 \{x_3,x_2\}&=&\frac1{\nu}[x_3,x_2]=a x_3 x_2+c x_3^2+d x_3.
\eea
Hence, the classical limit yields a Poisson algebra structure on $\C[ x_1,x_1^{-1},x_2,x_3,x_3^{-1} ]$ with the bracket being a sum of four compatible quadratic Poisson brackets, with coefficients $a,b,c,d$. If $a\ne 0$, the parameter  $d$ can be eliminated by a shift of the variable $x_2\to x_2-da^{-1}$. This Poisson bracket has rank two. A Casimir element (i.e., a generator of the Poisson centre) is given by:
\[ 
C =\lim\limits_{\nu\to 0}z=x_1^{-1}(a x_2+b x_1 +c x_3+d)x_3^{-1}.
\]
It is easy to verify that the automorphism $\T$ (\ref{Phi}) is a Poisson map: \[ \T(\{a,b\})=\{\T(a),\T(b)\},\quad \T(C )=C +1.\]

\subsection{Poisson structures on $N(4,\R)$.}
In this section, we list the Poisson brackets that arise as classical limits of the quantisation ideals discussed in Section~\ref{seq-class}. These define Poisson algebra structures on the nilpotent group  $N(4,\R)$, with elements parametrised as
\[
 \left(\begin{array}{llll}
 1 &s_1&s_3&s_6\\
 0&1&s_2&s_5\\
 0&0&1&s_4\\
 0&0&0&1
       \end{array}
\right).
\]
The mutations $\Phi_{i,j,k},\ (i,j,k)\in S$ are Poisson automorpisms of these algebras.
\medskip

\noindent
{\bf Classical limit of $I_0$}.\\
In the case of toric ideal $I_0$, assuming $\omega_k=1+\nu a_k$, in the limit $\nu\to 0$ we obtain  the following Poisson bi-vector $\pi^{(0)}_{i,j}=\{s_i,s_j\}$ with entries:
\[ 
\begin{array}{lll}
 \{s_2,s_1\}= a_1 s_1 s_2,\ &\{s_3,s_1\}= a_1 s_1 s_3,\ &\{s_3,s_2\}= -a_1 s_2 s_3,\\ \{s_4,s_1\}= \left(a_2-a_1\right) s_1 s_4,\ &\{s_4,s_2\}= \left(a_1-a_2+a_3\right) s_2 s_4,\ &\{s_4,s_3\}= a_3 s_3 s_4,\\ \{s_5,s_1\}= a_2 s_1 s_5,\ &\{s_5,s_2\}= \left(a_1-a_2+a_3\right) s_2 s_5,\ &\{s_5,s_3\}= \left(a_1+a_3\right) s_3 s_5,\\ \{s_5,s_4\}= -\left(a_1-a_2+a_3\right) s_4 s_5,\ &\{s_6,s_1\}= a_2 s_1 s_6,\ &\{s_6,s_2\}= \left(a_3-a_2\right) s_2 s_6,\\ \{s_6,s_3\}= a_3 s_3 s_6,\ &\{s_6,s_4\}= -a_3 s_4 s_6,\ &\{s_6,s_5\}= -a_2 s_5 s_6 .
\end{array} 
\]
The Poisson bi-vector  $\pi^{(0)}_{i,j}$ has rank 2 if at least one of the coefficients $a_i$ is non-zero. It is a linear combination of three compatible Poisson bi-vectors with the coefficients $a_1,a_2,a_3$. It represents a three-Hamiltonian structure on  $N(4,\R)$.

\medskip

\noindent
{\bf Classical limit of $I_1$}.\\
The ideal $I_1$ is a deformation of the toric ideal $I_0$. Assuming $\omega_i =1+\nu a_i $, we obtain in the limit $\nu\to 0$  the following Poisson bivector. For $j<i< 6$, the components coincide with the toric bivector $\pi^{(1)}_{i,j}=\pi^{(0)}_{i,j}$,  while the additional components involving  $ s_6 $ are:
\[
\begin{array}{l}
 \{s_6,s_1\}=a_2 s_1 s_6 -a_2 b_1 s_1^2+\left(a_1-a_2\right) b_2 s_1 s_2 - a_1 b_3 s_1 s_4,\\ \{s_6,s_2\}= \left(a_3-a_2\right) s_2 s_6 + \left(a_2-a_3\right) b_2 s_2^2-\left(a_1-a_2+a_3\right) b_1 s_1 s_2 + a_1 b_3 s_2 s_4,\\ \{s_6,s_3\}= a_3 s_3 s_6-\left(a_1+a_3\right) b_1 s_1 s_3 + \left(a_1-a_3\right) b_2 s_2 s_3,\\ \{s_6,s_4\}=-a_3 s_4 s_6+ a_3 b_3 s_4^2 + \left(a_1-a_2+a_3\right) b_1 s_1 s_4 + \left(a_2-a_1\right) b_2 s_2 s_4,\\ \{s_6,s_5\}= -a_2 s_5 s_6 - \left(a_1-2 a_2+a_3\right) b_2 s_2 s_5 + \left(a_1+a_3\right) b_3 s_4 s_5
\end{array}  .
\]
The Poisson bi-vector  $\pi^{(1)}_{i,j}$ has rank two for any choice of   parameters, except when $a_1=a_2=a_3=0$, in which case $\pi^{(1)}=0$. It is a non-trivial deformation of the toric Poisson bivector  $\pi^{(0)}$.

Moreover,  $\pi^{(1)}_{i,j}$ can be written as a linear combination of three compatible Poisson bivectors, with coefficients $a_1,a_2$ and $a_3$. Each of these bivectors  in turn is a linear combination of three compatible bivectors with coefficients $b_1,b_2$ and $b_3$.
\medskip

\noindent
{\bf Classical limit of $I_2$}.\\
In the case of   ideal $I_ 2$, assuming $\omega =1+\nu a $ and $C_i=\nu b_i$, in the limit $\nu\to 0$ we obtain  the following Poisson bi-vector $\pi^{(2)}_{i,j}=\{s_i,s_j\}$ with entries:
\[ 
\begin{array}{ll}
 \{s_2,s_1\}= a s_1 s_2,\ &\{s_3,s_1\}= a s_1 s_3,\\ 
 \{s_3,s_2\}= -a s_2 s_3,\ &\{s_4,s_1\}= 0,\\
 \{s_4,s_2\}= a s_2 s_4,\ &\{s_4,s_3\}= a s_3 s_4,\\
 \{s_5,s_1\}= a s_1 s_5,\ &\{s_5,s_2\}= a s_2 s_5,\\ 
 \{s_5,s_3\}= 2 a s_3 s_5, &\{s_5,s_4\}= -a s_4 s_5,
\end{array}
\]

\[
\begin{array}{l}
    \{s_6,s_1\} = a s_1 s_6 + b_1 s^2_1 + (b_3 - b_2) s_1 s_2 - b_5 s_1 s_4,\\
    \{s_6,s_2\} = b_1 s_1 s_2 + b_2 s^2_2 + b_5 s_2 s_4,\\ \{s_6,s_3\} = a s_3 s_6 + 2 b_1 s_1 s_3 + b_3 s_2 s_3,\\ \{s_6,s_4\} = -a s_4 s_6 - b_1 s_1 s_4 + b_4 s_2 s_4 + b_5 s^2_4,\\
    \{s_6,s_5\} = -a s_5 s_6 + (b_2 + b_5) s_2 s_5 + 2 b_5 s_4 s_5.
\end{array}
\]

The Poisson bi-vector  $\pi^{(2)}_{i,j}$ has rank 4.

\medskip

\noindent
{\bf Classical limit of $I_3$}.\\
In the case of   ideal $I_ 3$, assuming $\omega =1+\nu a $ and $C_i=\nu b_i$, in the limit $\nu\to 0$ we obtain  the following Poisson bi-vector $\pi^{(3)}_{i,j}=\{s_i,s_j\}$ with entries:
\[ 
\begin{array}{ll}
 \{s_2,s_1\}= a s_1 s_2,\ &\{s_3,s_1\}= a s_1 s_3,\\ 
 \{s_3,s_2\}= -a s_2 s_3,\ &\{s_4,s_1\}= 0,\\
 \{s_4,s_2\}= 2 a s_2 s_4,\ &\{s_4,s_3\}= 2 a s_3 s_4,\\
 \{s_5,s_1\}= a s_1 s_5,\ &\{s_5,s_2\}= 2 a s_2 s_5,\\ 
 \{s_5,s_3\}= 3 a s_3 s_5, &\{s_5,s_4\}= -2 a s_4 s_5,
\end{array}
\]

\[
\begin{array}{l}
    \{s_6,s_1\} = a s_1 s_6 + b_1 s^2_1 + b_2 s_1 s_4,\\
    \{s_6,s_2\} = a s_2 s_6 + 2 b_1 s_1 s_2 + b_4 s^2_2 - b_2 s_2 s_4,\\ \{s_6,s_3\} = 2 a s_3 s_6 + 3 b_1 s_1 s_3 + b_3 s^2_2 + b_4 s_2 s_3,\\ \{s_6,s_4\} = -2 a s_4 s_6 - 2 b_1 s_1 s_4 - 2 b_2 s^2_4,\\
    \{s_6,s_5\} = -a s_5 s_6 + b_4 s_2 s_5 - 3 b_2 s_4 s_5.
\end{array}
\]
The Poisson bi-vector  $\pi^{(4)}_{i,j}$ has rank 4.

\medskip

\noindent
{\bf Classical limit of $I_4$}.\\
In the case of   ideal $I_ 4$, assuming $\omega =1+\nu a $ and $C_i=\nu b_i$, in the limit $\nu\to 0$ we obtain  the following Poisson bi-vector $\pi^{(4)}_{i,j}=\{s_i,s_j\}$ with entries:
\[ 
\begin{array}{ll}
 \{s_2,s_1\}= 2 a s_1 s_2,\ &\{s_3,s_1\}= 2 a s_1 s_3,\\ 
 \{s_3,s_2\}= -2 a s_2 s_3,\ &\{s_4,s_1\}= 0,\\
 \{s_4,s_2\}= a s_2 s_4,\ &\{s_4,s_3\}= a s_3 s_4,\\
 \{s_5,s_1\}= 2 a s_1 s_5,\ &\{s_5,s_2\}= a s_2 s_5,\\ 
 \{s_5,s_3\}= 3 a s_3 s_5, &\{s_5,s_4\}= -a s_4 s_5,
\end{array}
\]

\[
\begin{array}{l}
    \{s_6,s_1\} = 2 a s_1 s_6 + 4 b_1 s^2_1 - 2 b_3 s_1 s_4,\\
    \{s_6,s_2\} = -a s_2 s_6 + 2 b_1 s_1 s_2 + b_2 s^2_2 + 2 b_3 s_2 s_4,\\ \{s_6,s_3\} = a s_3 s_6 + 6 b_1 s_1 s_3 + b_2 s_2 s_3,\\ \{s_6,s_4\} = -a s_4 s_6 - 2 b_1 s_1 s_4 + b_3 s^2_4,\\
    \{s_6,s_5\} = -2 a s_5 s_6 + b_4 s^2_2 + b_2 s_2 s_5 + 3 b_3 s_4 s_5.
\end{array}
\]
The Poisson bi-vector  $\pi^{(4)}_{i,j}$ has rank 4.
\medskip

\begin{rem}
Considering the connection of our problem with the family of parametrisations of the group of unipotent matrices (Section \ref{sec_param}), we can say that the obtained Poisson structures are structures consistent with the re-parametrisations on the charts of the unipotent group.
%\red{This property is related to the problem of constructing a family of log-canonical brackets on the charts of the cluster manifold consistent with the mutations \cite{GSV}.}
\end{rem}

\begin{rem}
    Consider the classical limit in the case of the non-generic ideal \eqref{other-ideal} $P_{*}$ yields  the Poisson bi-vector $\pi^{(*)}$ of rank 4.
\end{rem}

   \[
  \pi^{(*)}:\quad   \begin{array}{lll}
        \{s_1, s_2\} =     s_1 s_2, &\{s_1, s_3\} =     s_1 s_3, &\{s_1, s_4\} = 2  s_1 s_4, \\ \{s_1, s_5\}  = 3  s_1 s_5, &\{s_1, s_6\} = 3   s_1 s_6,& \{s_2, s_3\} = -    s_2 s_3, \\ \{s_2, s_4\}  = - 2 s_2 s_4,& \{s_2, s_5\} = - 2 s_2 s_5,& \{s_2, s_6\} = -    s_2 s_6, \\ \{s_3, s_4\}  = 0, & \{s_3, s_5\} =     s_3 s_5,& \{s_3, s_6\} = -C s^2_4, \\ \{s_4, s_5\}  =  2 s_4 s_5,& \{s_4, s_6\} = 0,& \{s_5, s_6\} = -    s_5 s_6.
    \end{array}
\]
\subsection{ Poisson centres and commuting integrals}

In the case of the log-canonical Poisson brackets with the bivector $\pi^{(0)}$, corresponding to the toric case, the Casimir elements - i.e., functions lying in the Poisson centre - can be computed explicitly. These elements take the form of monomials
\[
 \cC=s_1^{\alpha_1} s_2^{\alpha_2}\cdots s_6^{\alpha_1}
\]
where the exponent vector $(\alpha_1,\alpha_2,\ldots,\alpha_6)\in\C^6$ lies in the kernel of $\pi^{(0)}$.

In our case, one can choose the following monomials as Casimir functions:
\begin{equation}\label{cCs}
 \cC_1=s_1^{a_1-a_2+a_3}s_2^{a_1-a_2} s_4^{a_1},\quad \cC_2=s_1^{-1}s_2^{-1}s_3,\quad \cC_3=s_2^{-1}s_4^{-1}s_5,\quad \cC_4=s_1^{-1}s_2^{-1}s_4^{-1}s_6.
\end{equation}
These Casimir functions Poisson commute with all elements of the algebra and generate the Poisson centre of the Poisson algebra. They determine the symplectic foliation of the corresponding Poisson manifold  and are invariant under the Hamiltonian flows generated by any regular function in the algebra.

The mutations $\Phi_{ijk}$, for $(i,j,k)\in S$,  are Poisson maps for the Poisson algebras described in the previous section.
Moreover, the action of these Poisson maps, on the Casimir functions $\cC_i$ is affine-linear:
\begin{equation}\label{mutC}
 \begin{array}{ll}
  \Phi_{132} (\cC_i)=\cC_i+\delta_{i,2},&  \Phi_{254} (\cC_i)=\cC_i+\delta_{i,3}\\
   \Phi_{364} (\cC_i)=\cC_i+\delta_{i,4}\cC_2,&  \Phi_{165} (\cC_i)=\cC_i+\delta_{i,4}\cC_3,
 \end{array}
\end{equation}
where $\delta_{ij}$ denotes the Kronecker delta symbol. They equip the symplectic foliation with invariant lattice structure which we refer to as the {\em symplectic lattice}.

The variables
$s_1,s_2$ and $s_4$ remain invariant under the action of these maps. Any regular function of these variables, when taken as a Hamiltonian, generates integrable dynamics on the symplectic lattice, which is preserved by the Poisson maps.

A similar picture arises in the case of the deformation $\pi^{(1)}$
of the Poisson bivector $\pi^{(0)}$. In these case the Casimir elements $\cC_1,\cC_2$ and $\cC_3$ remain undeformed, while the element $\cC_4$ acquires a deformation:
\[\cC_4^*=s_1^{-1}s_2^{-1}s_4^{-1}(s_6-b_1 s_1-b_2 s_2-b_3 s_4).\]
The mutation rules (\ref{mutC}) remain unchanged under this deformation.

The cases corresponding to Poisson bivectors $\pi^{(2)},\pi^{(3)},\pi^{(4)}$ differ from the situation described above. While each of these bivectors arises as a deformation of a  particular case of the log-canonical Poisson bracket $\pi^{(0)}$, the deformation generally changes the structure in a significant way. In particular, it reduces the rank of the Poisson bivector to two, thus changing the symplectic foliation of the underlying manifold.

The Casimir functions generating the Poisson centre in these cases are:
\[
 \pi^{(2)} : \ \cC_2,\ \cC_3;\quad \pi^{(3)}: \ \cC_1'=s_1^2s_4,\ \cC_3;\quad \pi^{(4)}: \ \cC_1''=s_1s_4^2,\ \cC_2,
\]
where $\cC_2,\cC_3$ are as defined above in~\eqref{cCs}, and $\cC_1',\cC_1''$ arise from $\cC_1$ if we take into account the relations (\ref{c3}),(\ref{c4}). The functions $\cC_1',\cC_1''$ are mutation invariant.

In these cases, the symplectic leaves are four-dimensional, so to define a Liouville integrable Hamiltonian system, we must find two Poisson commuting functions. By a straightforward calculation we get:

\begin{prop}
The following pairs of functions Poisson commute with respect to the corresponding Poisson brackets:
\begin{enumerate}
  \item The functions $\cC_1'$ and $\cC_1''$ Poisson commute with respect to the bivector $\pi^{(2)}$.
  \item The functions $\cC_1''$ and $\cC_2$ Poisson commute with respect to the bivector $\pi^{(3)}$.
  \item The functions $\cC_1'$ and $\cC_3$ Poisson commute with respect to the bivector $\pi^{(4)}$.
 \end{enumerate}
\end{prop}

\section{Conclusion}

The main results of this paper concern quantum reductions of a noncommutative reparametrisation map on the unipotent group $N(4,\cA)$, where $\cA=\C\left<x_1,\ldots, x_6\right>$ is the free associative algebra. This construction yields quantum solutions to the Zamolodchikov tetrahedron equation.

\begin{itemize}
\item Using the method of quantisation ideals, we construct several families of associative algebras $\cA_{\cI}$ that are PBW deformations of the polynomial ring $\C[x_1,\ldots, x_6]$.  These deformations admit a well-defined quantum reductions of the re-parametrisation map to the unipotent group $N(4,\cA_{\cI})$, thereby providing quantum solutions to the Zamolodchikov tetrahedron equation.

\item   We examine the classical limit of these associative algebras and obtain a family of Poisson bivectors on the space of the unipotent group $N(4,\R)$, invariant under the re-parametrisation maps. Analogous problems have been extensively studied in the setting of cluster varieties, both in the Poisson and quantum contexts; see, for example, \cite{GSV}, \cite{BZ}.

\item We identify Hamiltonian integrable systems that are consistent with the mutation maps. Their solutions yield continuous symmetry groups acting on the parametrisations of the unipotent group. We expect these systems admit quantisations that respect mutation invariance.
\end{itemize}

We anticipate that the methods developed here will have broader applicability to a wide class of discrete dynamical systems with algebraic structure. In particular, we expect generalisations to the Lusztig variety, to mutation dynamics in electrical network models, to the Ising model, and to an expanding family of examples within the theory of cluster manifolds.

\section*{Acknowledgements}
The work of M.С. is an output of a research project implemented as part of the Basic Research Program at the National Research University Higher School of Economics (HSE University).
The work of D.T. was carried out within the framework of a development programme for the Regional Scientific and Educational Mathematical Center of Yaroslavl State University, with financial support from the Ministry of Science and Higher Education of the Russian Federation (Agreement on provision of subsidy from the federal budget No. 075-02-2025-1636).

\end{document}